\documentclass{LMCS}

\usepackage{enumerate}
\usepackage{hyperref}
\usepackage{natbib}
\usepackage{haskell}
\usepackage{stmaryrd}
\usepackage{undertilde}
\usepackage{enumerate,hyperref}

\theoremstyle{plain}

\DeclareMathOperator{\rng}{\it rng}
\DeclareMathOperator{\dom}{\it dom}
\newcommand{\kw}[1]{\mbox{\fontfamily{cmr}\fontseries{b}\selectfont#1}}
\newcommand{\LET}{\kw{let}\,}
\newcommand{\LETREC}{\kw{letrec}\,}
\newcommand{\IN}{\,\kw{in}\,}
\newcommand{\CASE}{\kw{case}\,}
\newcommand{\OF}{\,\kw{of}\,}
\newcommand{\alt}{\;\:\vert\;\:}
\newcommand{\arr}{\rightarrow}
\newcommand{\strict}{\mbox{\em strict}}
\newcommand{\linear}{\mbox{\it linear}}
\newcommand{\fv}{\mbox{\it fv}}
\newcommand{\fn}{\mbox{\it fn}}
\newcommand{\E}{\mathcal{E}}
\newcommand{\R}{\mathcal{R}}
\newcommand{\D}{\mathcal{D}}
\newcommand{\econ}[1]{\E\langle {#1} \rangle}
\newcommand{\con}[1]{\R\langle {#1} \rangle}
\newcommand{\eval}{\mapsto}
\newcommand{\eqdef}{\overset{\mbox{\tiny def}}{=}}
\newcommand{\adr}[2]{\D_{app}({#1})_{{#2},{\it G},\rho}}
\newcommand{\dr}[2]{\D\llbracket {#1} \rrbracket_{{#2},{\it G},\rho}}
\newcommand{\drp}[2]{\D\llbracket {#1} \rrbracket_{{#2},{\it G},\rho'}}
\newcommand{\drg}[2]{\D\llbracket {#1} \rrbracket_{{#2},{\it G}',\rho}}
\newcommand{\asdr}{\D_{app}\((\)}
\newcommand{\sdr}{\D\(\llbracket\)}
\newcommand{\edr}{\(\rrbracket\)}
\newcommand{\aedr}{\()\)}
\newcommand{\homemb}{\trianglelefteq}
\newcommand{\opapp}{\,\utilde{\sqsubset}\,}
\newcommand{\bopapp}{\,\utilde{\sqsupset}\,}
\newcommand{\opeq}{\cong}
\newcommand{\costeq}{\trianglelefteq \trianglerighteq}
\newcommand{\costlefteq}{\trianglelefteq}
\newcommand{\costrighteq}{\trianglerighteq}

\def\doi{6 (3:5) 2010}
\lmcsheading%
{\doi}
{1--39}
{}
{}
{Jun.~\phantom01, 2010}
{Aug.~17, 2010}
{}   

\begin{document}

\title[Positive Supercompilation]{Positive Supercompilation \\ for a Higher-Order Call-By-Value Language}

\author[P.~A.~Jonsson]{Peter A. Jonsson}	
\address{Lule{\aa} University of Technology \\
Department of Computer Science and Electrical Engineering}
\email{\{pj,nordland\}@csee.ltu.se}  

\author[]{Johan Nordlander}	
\address{\vskip-6 pt}





\begin{abstract}
  \noindent Previous deforestation and supercompilation algorithms may introduce
  accidental termination when applied to call-by-value programs. This
  hides looping bugs from the programmer, and changes the behavior of
  a program depending on whether it is optimized or not. We present a
  supercompilation algorithm for a higher-order call-by-value language
  and prove that the algorithm both terminates and preserves
  termination properties. This algorithm utilizes strictness
  information to decide whether to substitute or not and compares
  favorably with previous call-by-name transformations.
\end{abstract}

\maketitle

\section{Introduction}\label{S:one}
Intermediate data structures such as lists allow functional
programmers to write clear and concise programs, but carry a cost at
run-time since additional heap cells need to be both allocated and
garbage collected. Both deforestation
\citep{Wadler:1990:Deforestation} and supercompilation
\citep{Soerensen:1996:Positive} are automatic program transformations
which remove many of these intermediate structures. In a call-by-value
context these transformations are unsound, however, as they might hide
infinite recursion from the programmer. Consider the program
\begin{equation*}
(\lambda x.y)\, (\mbox{{\it fac}} \,z).
\end{equation*}

This program could loop, if the value of {\em z}
is negative. Applying Wadler's deforestation algorithm to the program will
result in \<y\>, which is sound under call-by-name or
call-by-need. Under call-by-value the non-termination in the original
program has been removed, and hence the meaning of the program has
been altered by the transformation.

This is unfortunate since removing intermediate structures in a call-by-value language is perhaps even
more important than in a lazy language since the entire intermediate
structure has to remain in memory during the computation. 

\citet{Ohori:2007:Lightweight} saw this need and presented a very
elegant algorithm for call-by-value languages that removes
intermediate structures.  Their algorithm sacrifices some
transformational power for algorithmic simplicity. In this article we explore a
different part of the design space: a more powerful transformation at
the cost of some algorithmic complexity.  The outcome is a
meaning-preserving supercompiler for pure call-by-value languages in
general, together measurements from an implementation in a compiler for Timber \citep{Timber}, a pure call-by-value language.

Our current work is a necessary first step towards supercompiling
impure call-by-value languages, of which there are many 
available today. Well-known examples are OCaml \citep{Ocaml}, Standard
ML \citep{SML97} and F\# \citep{Syme:2008:The}. Considering that F\#
is currently being turned into a product, it is quite likely that
strict functional languages will be even more popular in the future.

One might think that our result should be easy to obtain by modifying
a call-by-name algorithm to simply delay beta reduction until every
function argument has been specialized to a value. However, it turns
out that this strategy misses even simple opportunities to remove
intermediate structures. That is, eager specialization of function arguments
risks destroying {\em fold} opportunities that might otherwise appear,
something which may prohibit complexity improvements to the resulting
program.

The novelty of our supercompilation algorithm is that it concentrates
all call-by-value dependencies to a single rule that relies on the
result from a separate strictness analysis for correct behavior. In
effect, our transformation delays transformation of function arguments past
inlining, much like a call-by-name scheme does, although only as far
as allowed by call-by-value semantics.  The result is an algorithm
that is able to improve a wide range of illustrative examples like the
existing algorithms do, but without the risk of introducing artificial
termination.

The specific contributions of our work are:

\begin{enumerate}[$\bullet$]
\item We provide an algorithm for positive supercompilation including
  folding, for a strict and pure higher-order functional language
  (Section \ref{sec:pscp}).
\item We prove that the algorithm terminates and preserves the
  semantics of the program (Section \ref{sec:correctness}).
\item We show preliminary benchmarks from an implementation in the
  Timber compiler (Section \ref{sec:benchmarks}).
\end{enumerate}

We start out with some examples in Section \ref{sec:examples} to give
the reader an intuitive feel of how the algorithm behaves. Our
language of study is defined in Section \ref{sec:lang}, right before
the technical contributions are presented.

This article is an extended and improved version of a paper presented
at POPL 2009 \citep{Jonsson:2009:Positive}. As well as clarifying a
number of the examples and proofs, we give an improved
formulation of \<\asdr{}\aedr\> presented in
Section~\ref{sec:apprule}, and make a small change to how
let-expressions are handled by the driving algorithm.

\section{Examples}
\label{sec:examples}

\citet{Wadler:1990:Deforestation} uses the example \<append
 (append xs ys) zs\> and shows that his
deforestation algorithm transforms the program so that it saves one
traversal of the first list, thereby reducing the complexity
from $2|xs| + |ys|$ to $|xs| + |ys|$. 

If we na\"{i}vely change Wadler's algorithm to call-by-value semantics
by eagerly attempting to transform arguments before attacking the
body, we do not achieve this improvement. The definition for append is:
\begin{haskell}
append xs ys & = & \hscase{xs}{[] \to ys\\ (x':xs') \to x':append xs' ys} \\
\end{haskell}
\noindent and we give an example of a hypothetical call-by-value variant of
Wadler's deforestation algorithm that attacks arguments first:
 
\begin{tabular}{lp{0.9\textwidth}l}
~ \\
& \<
 append (append xs' ys') zs' 
\> \\
~ \\
& Inlining the body of the inner {\it append}
 and then pushing down the outer call into each branch gives \\
~ \\
 &
\<
 \hscase{xs'}{[] \to\!\!  append\,ys' zs' \\ (x:xs) \to\!\! append
 (x:append\, xs\, ys') zs'}
\>\\
~ \\
& Transformation of the first branch will create a new function \<h_1\> that is isomorphic to {\it append}, and call it. The second branch contains an embedding of the initial expression and blindly transforming it will lead to non-termination of the transformation algorithm.
  One must therefore split this expression in two parts: the subexpression \<x:append xs ys'\> which we call \<h_2\>, and the outer expression \<append z zs'\> where \<z\> is fresh. Continuing with \<x:append xs ys'\> and inlining \<append\> gives \\
~ \\
 &
\<
x:\hscase{xs}{[] \to ys' \\ (x':xs') \to x':append xs' ys' }
\>\\
~ \\
& The second branch contains a renaming of the expression we named \<h_2\>, so we simply replace it with a call to \<h_2\>. Moving back to \<append z zs'\> we notice that this expression is a renaming of the one called \<h_1\>, so we replace it with the call \<h_1 z zs'\> \\
~ \\
\end{tabular}

\noindent Assembling the pieces gives us the final result:
\begin{haskell*}
 & &  \LETREC h_1 xs ys = \hscase{xs}{[] \to ys \\ (x':xs') \to x':h_1 xs' ys} \\
& & \qquad \quad \quad \quad\quad \quad \quad \quad \quad h_2 x xs ys = x:\hscase{xs}{[] \to ys \\(x':xs') \to h_2 x' xs' ys} \\
 & & \IN \hscase{xs'}{[] \to h_1 ys' zs' \\ (x:xs) \to h_1 (h_2 x xs ys')
  zs'}  \\
\end{haskell*}
Notice that the intermediate structure from the input program is still there after
the transformation, and the complexity is still $2|xs| + |ys|$! This
can be compared to how the same example is transformed by Wadler's
algorithm as shown in Figure \ref{fig:wadlerappend}. The reason our
hypothetical call-by-value algorithm failed to improve the program is
that it had to split expressions too early during the transformation,
thereby preventing fold opportunities that occur in a
call-by-name setting.

\begin{figure*}
\begin{tabular}{lp{0.85\linewidth}l}
& \<
 append (append xs' ys') zs' 
\> \\
~ \\
& Naming the first expression \<h_1\> and inlining both occurrences of {\it append} gives \\
~ \\
 &
\<
 \hsccase{(\hscase{xs'}{[]\to ys' \\ (x_1:xs_1) \to x_1:append xs_1 ys') \OF}}{[] \to\!\!  zs' \\ (x:xs) \to\!\! x:append xs zs'} \>\\
~ \\
& Pushing down the outer case-expression into both branches of the inner one and reducing the resulting case-expression of a known constructor leads to \\
~ \\
 &
\<
 \hscase{xs'}{[]\to \hscase{ys'}{[] \to\!\!  zs' \\ (x:xs) \to\!\! x:append xs zs'} \\ (x_1:xs_1) \to  x_1:append (append xs_1 ys') zs'} \>\\
~ \\
& Transform each branch separately. Transformation of the second branch in the first branch will create a new function \<h_2\> that is isomorphic to {\it append}, and the second branch of the outer case is a renaming of our initial expression called \<h_1\>. Assembling all pieces yields the following result: \\
~ \\
 &\<\LETREC h_1 xs ys zs = \hscase{xs}{[] \to \hscase{ys}{[] \to zs\\(y':ys') \to y':h_2 ys' zs} \\ (x':xs') \to x':h_1 xs' ys zs} \>\\
 & \<\qquad \quad \quad \quad\quad \quad \quad \quad \quad  h_2 xs ys = \hscase{xs}{[] \to ys\\ (x':xs') \to x':h_2 xs' ys}\> \\
&  \<\IN h_1 xs' ys' zs' \> \\
\end{tabular}
\caption{Wadler's algorithm transforming {\it append (append xs' ys') zs'}}
\label{fig:wadlerappend}
\end{figure*}

However, changing the call-by-value algorithm to do the exact opposite --- that is, carefully delaying the
transformation of arguments to a function past the inlining of its body, but only as far as strictness allows
--- actually leads to the same result that Wadler obtains with \<append (append xs ys)  zs\>. This is a key observation
for obtaining deforestation under call-by-value without
altering the semantics, and our transformation exploits it.

Except for the fundamental reliance on strictness analysis, which is necessary to preserve
semantics, our transformation shares many of its rules with Wadler's
algorithm. The transformation that is commonly referred to as
case-of-case is crucial for our transformation, just like it is for a
call-by-name algorithm.  The case-of-case transformation is useful
when a case-expression appears in the head of another case-expression,
in which case the outer case context is duplicated and pushed into
all branches of the inner case-expression.  Our transformation also contains
rules that correspond to ordinary evaluation which eliminate case-expressions that have a known constructor in their head or adds two
primitive numbers. The mechanism that ensures termination basically
looks for ``similar'' terms to ones that have already been
transformed, and if a similar term is encountered, the transformation
will stop and split the term into smaller terms that are
transformed separately. The remaining rules of our transformation simply shifts
 focus to the proper subexpression and ensures that the
algorithm does not get stuck.

We claim that our transformation compares favorably with previous
call-by-name transformations, and we now proceed with demonstrating the
transformation on some common examples. The results of the
transformation on these examples are identical to the results of
Wadler's algorithm \citep{Wadler:1990:Deforestation}.  

This does not
hold in general, a counter-example is the transformation of the expression \<zip (map
 f xs) (map g ys)\> where Wadler's algorithm will eliminate both
intermediate structures and our transformation will only eliminate the first
intermediate structure. 

Our first example is transformation of \<sum (map square ys)\>, where the referenced functions are defined as:
\begin{haskell}
square x & = &\,\, x*x \\
map f xs & = & \hscase{xs}{[] \to ys \\ (x:xs) \to f x:map f xs} \\
sum xs & = &\hscase{xs}{[] \to 0 \\ (x:xs) \to x + sum xs}\\
\end{haskell}

We start our transformation by allocating a new fresh function name
\<h_0\> to the expression \<sum (map square ys)\>, inlining the body of \<sum\> and
substituting \<map square ys\> into the body of \<sum\>:
\begin{haskell*}
\hscase{map square ys}{[] \to 0\\ (x':xs') \to x' + sum xs' }
\end{haskell*}

\noindent After inlining \<map\> and substituting the arguments into
the body the result becomes:
\begin{haskell*}
\hsccase{(\hscase{ys}{[] \to []\\ (x':xs') \to (square x'):map square xs') \hskwd{of}}}{   [] \to 0 \\(x':xs') \to x' + sum xs'}
\end{haskell*}

We duplicate the outer case in each of the inner case branches, using
the expression in the branches as head of that
case-expression. Continuing the transformation on each branch with ordinary
reduction steps yields: 
\begin{haskell*}
\hscase{ys}{[] \to 0\\ (x':xs') \to square x' +  sum (map square xs') }
\end{haskell*}

At this point we inline the body of the first {\it square} occurrence and observe that
the second parameter to \<(+)\> is similar to the expression we started
with and therefore we replace it with \<h_0 xs'\>. The result of our
transformation is \<h_0 ys\>, with \<h_0\> defined as:
\begin{haskell*}
h_0 ys &=& \hscase{ys}{[] \to 0\\ (x':xs') \to x'*x' + h_0 xs'} \\
\end{haskell*}

This new function only traverses its input once, and no intermediate
structures are created. If the expression {\em sum (map square xs)}
or a renaming of it is detected elsewhere in the input, a call to
\<h_0\> will be inserted instead.

The work by \citet{Ohori:2007:Lightweight} cannot fuse two successive
applications of the same function, nor mutually recursive
functions. We show in the next two examples that our
transformation can handle these cases. We need the following new
function definitions:
\begin{haskell*}
mapsq xs & = &\hscase{xs}{ [] \to []\\ (x':xs') \to (x'*x'):mapsq xs' } \\
f xs & = &\hscase{xs}{ [] \to []\\ (x':xs') \to (2*x'):g xs'} \\
g xs & = &\hscase{xs}{ [] \to []\\ (x':xs') \to (3*x'):f xs'} \\
\end{haskell*}

Transforming \<mapsq (mapsq xs)\> will inline the outer
\<mapsq\>, substitute the argument in the function body and  inline
the inner call to \<mapsq\>:
\begin{haskell*}
\hsccase{(\hscase{xs}{[] \to []\\ (x':xs') \to (x'*x'):mapsq
  xs') \hskwd{of}}}{   [] \to [] \\(x':xs') \to (x'*x'):mapsq xs'  }
\end{haskell*}

As previously, we duplicate the outer case in each of the inner case
branches, using the expression in the branches as head of that
case-expression. Continuing the transformation on each branch by ordinary
reduction steps yields:

\begin{haskell*}
\hscase{xs}{[] \to []\\ (x':xs') \to  (x'*x'*x'*x'):mapsq (mapsq xs')  } \\
\end{haskell*}

Here we encounter a similar expression to what we started with, and
create a new function \<h_1\>. The final result of our
transformation is \<h_1 xs\>, with the new residual function \<h_1\>
that only traverses its input once defined as:
\begin{haskell*}
h_1 xs & = & \hscase{xs}{[] \to []\\ (x':xs') \to (x'*x'*x'*x'):h_1 xs'  } \\
\end{haskell*}

For an example of transforming mutually recursive functions, consider
the transformation of \<sum (f xs)\>. Inlining the body of \<sum\>, substituting
its arguments in the function body and inlining the body of \<f\>
yields:
\begin{haskell*}
\hsccase{(\hscase{xs}{[] \to []\\ (x':xs') \to (2*x'):g xs') \hskwd{of}}}{   [] \to 0 \\(x':xs') \to x' + sum xs'  }
\end{haskell*}

\noindent We now move down the outer case into each branch, and
perform reductions until we end up with: 
\begin{haskell*}
\hscase{xs}{[] \to 0 \\ (x':xs') \to (2*x') + sum (g xs')} \\
\end{haskell*}

We notice that unlike in previous examples, \<sum (g xs')\> is not
similar to what we started transforming and we can therefore continue the transformation. For space reasons, we focus
on the transformation of the rightmost expression in the last branch,
\<sum (g xs')\>, while keeping the functions already seen in
mind. We inline the body of \<sum\>, perform the substitution of its
arguments and inline the body of \<g\>:
\begin{haskell*}
\hsccase{(\hscase{xs'}{[] \to []\\ (x'':xs'') \to (3*x''):f  xs'') \hskwd{of}}}{   [] \to 0 \\(x':xs') \to x' + sum xs'  }
\end{haskell*}

\noindent We now move down the outer case into each branch, and perform
reductions:
\begin{haskell*}
 \hscase{xs'}{ [] \to 0 \\(x'':xs'') \to (3*x'') + sum (f xs'')   }
\end{haskell*}

We notice a familiar expression in \<sum (f xs'')\>, and fold when
reaching it. Combining the fragments together gives a new function
\<h_2\>:
\begin{haskell*}
h_2 xs & = & \hstcase{xs}{[] \to 0 \\ (x':xs') \to (2*x') + \hstcase{xs'}{ [] \to  0\\(x'':xs'') \to   (3*x'') + h_2 xs''} }\\
\end{haskell*}

The new function \<h_2\> consumes a list and returns a number, so our
algorithm has eliminated the intermediate list between \<f\> and \<sum\>.

\citet{Kort:1996:Deforestation} studied a ray-tracer written in
Haskell, and identified a critical function in the innermost loop of a
matrix multiplication, called \<vecDot\>:

\begin{haskell*}
vecDot xs ys &=& sum (zipWith (*) xs ys) 
\end{haskell*}

\noindent This is simplified by our positive supercompiler to:
\begin{haskell*}
vecDot xs ys &=& h_1 xs ys \\
h_1 xs ys &=& \hstcase{xs}{
              (x':xs') \to
                      \hstcase{ys}{
                       (y':ys') \to x' * y' + h_1 xs' ys'\\
                        \_ \to 0} \\
              \_ \to 0}
\end{haskell*}

The intermediate list between {\em sum} and
{\em zipWith} is transformed away, and the complexity is reduced from
$2|xs| + |ys|$ to $|xs| + |ys|$ (since this is matrix multiplication
$|xs| = |ys|$).

\section{Language}
\label{sec:lang}

Our language of study is a strict, higher-order functional language
with let-bindings and case-expressions. Its syntax for expressions,
values and patterns is shown in Figure \ref{fig:language}.

\begin{figure}[h]
\begin{tabular}{lll}
\multicolumn{3}{l}{Expressions} \\ \hline \\
  $e,\,f$ &::= & $n \alt x \alt g \alt f\, e \alt  \lambda x.e \alt k \, \overline{e} \alt e_1  \oplus  e_2 \alt \CASE e \,\OF \, \{ p_i \arr e_i \} $ \\ 
& $\alt$ & $\LET x = f\, \IN e \alt \LETREC g = v \IN e $ \\
\\
$p$ &::=& $n \alt k\, \overline{x}$ \\
\\
\multicolumn{3}{l}{Values} \\ \hline \\
  $v$ &::= & $n \alt \lambda x.e
    \alt k \, \overline{v}$ \\
\end{tabular}  
\caption{The language}
\label{fig:language}
\end{figure}

Here we let variables and constructor symbols be denoted by {\it x}
and {\it k}, respectively. The constructor symbols {\it k} range over
a set {\it K} and we also assume that there is a separate set
$\mathcal{G}$ of recursively defined function symbols, ranged over by
{\it g}. In what follows we will assume that the meaning of such
symbols is given by a recursive map {\it G} mapping symbols {\it g} to
their defined value.

The language contains integer values $n$ and arithmetic
operations $\oplus$, although these meta-variables can preferably be
understood as ranging over primitive values in general and arbitrary
operations on these. We let $+$ denote the semantic meaning of
$\oplus$.

A list of expressions $e_1 \ldots
e_n$ is abbreviated as $\overline{e}$, and a list of variables $x_1 \ldots x_n$ as
$\overline{x}$.

We denote the free variables of an expression {\em e} by $\fv(e)$, as
defined in Figure \ref{fig:fv}. Along the same lines we denote the
function names in an expression {\em e} as $\fn(e)$, defined in
Figure~\ref{fig:fn}.

\begin{figure}
$$
\begin{array}{lll}
\fv(x) & = & \{x\} \\
\fv(n) & = & \emptyset \\
\fv(g) & = & \emptyset \\
\fv(k\, \overline{e}) & = & \fv(\overline{e})  \\
\fv(\lambda x.e) & = & \fv(e) \backslash \{x\} \\
\fv(f\, \overline{e}) & = & \fv(f) \cup \fv(\overline{e}) \\
\fv(\LET x = e \IN f) & = & \fv(e) \cup (\fv(f) \backslash \{x\}) \\
\fv(\LETREC g = v \IN f) & = & \fv(v) \cup \fv(f) \\
\fv(\CASE e \OF \{p_i \arr e_i \}) & = & \fv(e) \cup (\bigcup (\fv(e_i)\backslash \fv(p_i))) \\ 
\fv(e_1 \oplus e_2) & = & \fv(e_1) \cup \fv(e_2) \\
\end{array}
$$
\caption{Free variables of an expression}
\label{fig:fv}
\end{figure}

\begin{figure}
$$
\begin{array}{lll}
\fn(x) & = & \emptyset \\
\fn(n) & = & \emptyset \\
\fn(g) & = & \{g\} \\
\fn(k\, \overline{e}) & = & \fn(\overline{e})  \\
\fn(\lambda x.e) & = & \fn(e) \\
\fn(f\, e) & = & \fn(f) \cup \fn(e) \\
\fn(\LET x = e \IN f) & = & \fn(e) \cup \fn(f) \\
\fn(\LETREC g = v \IN f) & = & (\fn(v) \cup \fn(f)) \backslash \{g\} \\
\fn(\CASE e \OF \{p_i \arr e_i \}) & = & \fn(e) \cup (\bigcup (\fn(e_i)) \\ 
\fn(e_1 \oplus e_2) & = & \fn(e_1) \cup \fn(e_2) \\
\end{array}
$$
\caption{Function names of an expression}
\label{fig:fn}
\end{figure}

We encode {\bf letrec} as an application containing {\it fix}, where {\it fix} is
defined as
\begin{equation*}
fix =\lambda f.f\,(\lambda n.fix\, f\,n)
\end{equation*}

\begin{defi} Letrec is defined as:
\begin{equation*}
\LETREC h = \lambda \overline{x}.e \IN e' \eqdef (\lambda h.e')\,(\lambda y.fix\,(\lambda h.\lambda \overline{x}.e)\,y)
\end{equation*}
\end{defi}

By defining letrec as syntactic sugar for other primitives we
introduce an implicit requirement that the right hand side of letrec
expressions must not contain any free variables except h. This is not a
limitation since functions that contain free variables can be lambda
lifted \citep{Johnsson:1985:Lambda} to the top level.

A program is an expression with no free variables and all function
names defined in {\it G}.  The intended operational
semantics is given in Figure \ref{fig:redsem}, where
$[\overline{e}/\overline{x}]e'$ is the capture-free substitution of
expressions $\overline{e}$ for variables $\overline{x}$ in $e'$. 

\begin{figure}
\begin{tabular}{lll}
\multicolumn{3}{l}{Reduction contexts} \\ \hline \\
$\E$  &::= & 
      $\boxempty
      \alt            \E\, e
      \alt            (\lambda x.e)\, \E
       \alt      k \, \overline{\E}
  \alt \E \oplus e \alt n \oplus \E 
        \alt  \CASE\, \E \, \OF \, \{ p_i \arr e_i \} \alt \LET x=\E \, \IN e$   \\
\\
\end{tabular} \\
\begin{tabular}{lllr}
\multicolumn{4}{l}{Evaluation relation} \\ \hline \\
$\econ{g}$ & $\eval$ & $\econ{v}$, if $(g, v) \in ${\it G}   & (Global) \\
$\econ{(\lambda x.e) \,v}$ & $\eval$ & $\econ{[v/x]e}$ &  (App) \\
$\econ{\LET x=v \IN e}$ & $\eval$ & $\econ{[v/x]e}$ & (Let) \\
 $\econ{\CASE k \, \overline{v} \OF \, \{ k_i\,\overline{x}_i \arr e_i \}}\!\!\!\!\!$ & $\eval$ & $\econ{[\overline{v}/\overline{x}_j]e_j}$, if $k = k_j$ & (KCase) \\
$\econ{\CASE n \OF \,\{ n_i \arr e_i \} }$ & $\eval$ & $\econ{e_j}$, if $n = n_j$ & (NCase) \\
 $\econ{n_1 \oplus n_2}$ & $\eval$ & $\econ{n}$, if $n = n_1 + n_2$  & (Arith) \\
\end{tabular}
\caption{Reduction semantics}
\label{fig:redsem}
\end{figure}

A reduction context $\E$ is a term containing a single hole, $\boxempty$,
which indicates the next expression to be reduced. The expression
$\econ{e}$ is the term obtained by replacing the hole in $\E$ with
$e$. $\overline{\E}$ denotes a list of terms with just a single hole,
evaluated from left to right.

If a variable appears no more than once in a term, that term is said
to be $\linear$ with respect to that variable. Like 
\citet{Wadler:1990:Deforestation}, we extend the definition slightly
for linear case-expressions: no variable may appear in both the head
and a branch, although a variable may appear in more than one
branch. For example, the definition of {\em append} is linear is linear with respect to {\it ys}, although {\it ys}
appears in both branches.

\section{Higher Order Positive Supercompilation}
\label{sec:pscp}

It is time to make the intuition developed in Section
\ref{sec:examples} more formal.  Our supercompiler is defined as a set
of rewrite rules that pattern-match on expressions. This algorithm is
called the {\em driving} algorithm, and is defined in Figure
\ref{fig:drivalg}. Three additional parameters appear as subscripts to
the rewrite rules: a driving context
$\R$, the set of global function definitions {\it G} and a memoization list $\rho$. The memoization list holds information about expressions already
traversed and is explained more in detail in Section
\ref{sec:apprule}. The driving context $\R$ is smaller than $\E$, and
is defined as follows:
\begin{equation*}
\R ::= \boxempty
      \alt            \R \, e
      \alt   	      \CASE \R \, \OF\, \{ p_i \arr e_i \}
      \alt \R \oplus e
      \alt e \oplus \R
\end{equation*}
Interestingly, this definition coincides with the evaluation contexts
for a call-by-name language.  The reason our transformation still preserves
a call-by-value semantics is that beta reduction (rule R9) results in
a let-binding, whose further specialization in rule R13 depends on
whether the body expression $f$ is strict in the bound variable $x$ or
not. 

Our let-rule (R13) might change the order of computations, but since
non-termination is commutative this does not matter in
practice. Supercompiling impure languages requires stronger conditions
for the let-rule, since expressions might contain effects other than
non-termination. The difficulty of supercompiling an impure language
is to find sufficient conditions that preserve soundness while still
allowing the maximum amount of reordering of expressions.

\begin{figure*}
\begin{tabular}{lllr}
$\dr{n}{\R}$  & = & $\con{n}$ & (R1) \\
$\dr{x}{\R}$ & = & $\con{x}$ & (R2) \\
$\dr{g}{\R}$ & = & $\adr{g}{\R}$ & (R3) \\
$\dr{k \, \overline{e}}{\boxempty}$ & = & $k \, \dr{\overline{e}}{\boxempty}$ & (R4) \\
$\dr{x \, \overline{e}}{\R}$ & = & $\con{x \,\dr{\overline{e}}{\boxempty}}$ & (R5) \\
$\dr{\lambda \overline{x}.e}{\boxempty}$ & = & $(\lambda \overline{x}.\dr{e}{\boxempty})$ & (R6) \\
$\dr{n_1 \oplus n_2}{\R}$ & = & $\dr{\con{n}}{\boxempty}$, where $n = n_1 + n_2$ & (R7) \\
 $\dr{e_1 \oplus e_2}{\R}$ & = & $\dr{e_1}{\boxempty} \oplus \dr{e_2}{\boxempty}$, if $e_1 \oplus e_2 = a$ & (R8) \\
  & & $\dr{e_2}{\con{e_1 \oplus \boxempty}}$, if $e_1 = n$ or $e_1 = a$ \\
  & & $\dr{e_1}{\con{\boxempty\oplus e_2}}$, otherwise \\
$\dr{(\lambda \overline{x}.f)\, \overline{e}}{\R}$ & = & $\dr{\LET \overline{x} = \overline{e} \IN f}{\R}$ & (R9) \\
$\dr{e\, e'}{\R}$ & = & $\dr{e}{\con{\boxempty \, e'}}$ & (R10) \\
$\dr{\LET x=n \IN f}{\R}$ & = & $\dr{\con{[n/x]f}}{\boxempty}$ & (R11) \\ 
$\dr{\LET x=y \IN f}{\R}$ & = & $\dr{\con{[y/x]f}}{\boxempty}$,  if $y$ not freshly generated & (R12) \\ 
$\dr{\LET x=e \IN f}{\R}$ & = & $\dr{\con{[e/x]f}}{\boxempty}$, if $x \in \strict(f)$ and  & (R13) \\ 
& & \hspace{9.4em}$x \in \linear(f)$ \\
   & & $\LET x = \dr{e}{\boxempty}\, \IN \dr{\con{f}}{\boxempty}$, otherwise & \\
$\dr{\LETREC g = v \IN e}{\R}$ & = & $\drg{\con{e}}{\boxempty}$,  where {\it G}' = {\it G} $\cup (g, v)$  & (R14) \\
$\dr{\CASE x\, \OF \{p_i \arr e_i \}}{\R}$ & = & $\CASE x\, \OF \{p_i \arr \dr{[p_i/x]\con{e_i}}{\boxempty} \}$ & (R15) \\
$\dr{\CASE k_j\, \overline{e} \OF \{k_i\,\overline{x}_i \arr e_i  \}}{\R}$ & = & $\dr{\con{\LET \overline{x}_j = \overline{e} \IN  e_j}}{\boxempty}$ & (R16)\\
$\dr{\CASE n_j\, \OF \{n_i \arr e_i \}}{\R}$ & = & $\dr{\con{e_j}}{\boxempty}$ & (R17)\\
$\dr{\CASE a\, \OF  \{p_i \arr e_i \} }{\R}$ & = & $\CASE \dr{a}{\boxempty} \OF  \{p_i  \arr \dr{\con{e_i}}{\boxempty} \} $  & (R18) \\
$\dr{\CASE e\, \OF  \{p_i \arr e_i \} }{\R}$ & = & $\dr{e}{\con{\CASE \boxempty \OF  \{p_i  \arr e_i \}}}$ & (R19) \\
$\dr{e}{\R}$  & = & $\con{e}$ & (R20) \\
\end{tabular}
\caption{Driving algorithm}
\label{fig:drivalg}
\end{figure*}

In principle, an expression $e$ is strict with regards to a variable
$x$ if evaluation of {\it e} eventually requires the value of $x$; in other words, if $e \eval \ldots
\eval \econ{x}$.  Such information is not computable in general, 
although call-by-value semantics allows for reasonably tight
approximations.  One such approximation is given in Figure
\ref{fig:strictvars}, where the strict variables of an expression $e$
are defined as all free variables of $e$ except those that only appear
under a lambda or not inside all branches of a case. 

\begin{figure}
$$
\begin{array}{lll}
\strict(x) & = & \!\!\!\! \{x\} \\
\strict(n) & = & \!\!\!\! \emptyset \\
\strict(g) & = & \!\!\!\! \emptyset \\
\strict(k\, \overline{e}) & = & \!\!\!\! \strict(\overline{e})  \\
\strict(\lambda x.e) & = & \!\!\!\! \emptyset \\
\strict(f\, e) & = & \!\!\!\! \strict(f) \cup \strict(e) \\
\strict(\LET x = e \IN f) & = & \!\!\!\! \strict(e) \cup (\strict(f) \backslash \{x\}) \\
\strict(\LETREC g = v \IN f) & = & \!\!\!\! \strict(f) \\
\strict(\CASE e \OF \{p_i \arr e_i \}) \!\!\!\!\! & = & \!\!\!\! \strict(e) \cup  
(\bigcap (\strict(e_i)\backslash \fv(p_i))) \\ 
\strict(e_1 \oplus e_2) & = & \!\!\!\! \strict(e_1) \cup \strict(e_2) \\
\end{array}
$$
\caption{The strict variables of an expression}
\label{fig:strictvars}
\end{figure}

There is an ordering between the driving rules; i.e., all rules must be tried in
the order they appear. Rule R10 is the default fallback case for
applications and rule R19 is the default fallback case for case
expressions. These rules extend the driving context $\R$ and
zoom in on the next expression to be driven. The program is turned
``inside-out'' by moving the surrounding context $\R$ into all
branches of the case-expression through rules R15 and R18. Rule R13 has
a similar mechanism for let-expressions.  Notice how the context is
moved out of the recursive call in rule R5, whereas rule R7
recursively applies the driving algorithm to the full new term
$\con{n}$, forcing a re-traversal of the new term in search for for further
reduction opportunities. Rule R12 is only allowed to match if the variable {\em y}
is not freshly generated by the
splitting mechanism described in the next section. Meta-variable {\em
  a} in rules R8 and R18 stands for an ``annoying'' expression;
i.e., an expression that would be further reducible were it not for a
free variable getting in the way.  The grammar for annoying
expressions is:
$$
\begin{array}{lll}
a &::=& x \alt n \oplus a \alt a \oplus n \alt a \oplus a \alt a\, \overline{e} \\
\end{array}
$$

Some expressions should be handled differently depending on
context. If a constructor application appears in an empty context,
there is not much we can do but to drive the argument expressions
(rule R4). On the other hand - if the application occurs at the head
of a case-expression, we may choose a branch on basis of the
constructor and leave the arguments unevaluated in the hope of finding
fold opportunities further down the syntax tree (rule R16).

The argumentation is analogous for lambda abstractions: if there is
a surrounding application context we perform a beta reduction, otherwise we proceed by driving the abstraction itself.

Notice that the primitive operations ranged over by $\oplus$ cannot
be unfolded and transformed like ordinary functions can. If the
arguments of a primitive operation are annoying, our transformation will
simply leave the primitive operation in place (rule R8).

If we had a perfect strictness analysis and could decide whether an
arbitrary expression will terminate or not, the only difference in
results between our transformation and a call-by-name counterpart would be
for the non-terminating cases. In practice, we have to settle for an
approximation, such as the simple analysis defined in Figure
\ref{fig:strictvars}. One might speculate whether the transformations
thus missed will have adverse effects on the usefulness of our
transformation in practice.  We believe we have seen clear indications that
this is not the case, and that the crucial factor is the
ability to inline function bodies irrespective of whether arguments
are values or not.

Our transformation always inlines functions unless the algorithm detects a
risk of non-termination. Supero
\citep[Sec. 3.2]{Mitchell:2008:ASupercompiler} has a more advanced
inlining strategy.

\subsection{Application Rule}

\label{sec:apprule}

In the driving algorithm rule R3 refers to \<\asdr{}\aedr\>,
defined in Figure~\ref{fig:adrivalg}.  \<\asdr{}\aedr\> can be inlined
in the definition of the driving algorithm, it is merely given a
separate name to improve the clarity of the presentation. Figure
\ref{fig:adrivalg} contains some new notation: we use $\sigma$ for a
variable to variable substitution and \<=\> for syntactic
equivalence of expressions.

\begin{figure*}
\begin{tabular}{llllr}
$\adr{g}{\R}$ &$\! =\, $& $\!\!h \, \overline{x}$ & if $\exists(h, e_1)\! \in \!\rho\,.\, \sigma e_1 =  \con{g}$  & (1)\\
 \multicolumn{4}{l}{\hspace{2em} where $\overline{x} = \sigma(\fv(e_1))$} \\ 
$\adr{g}{\R}$ &$\! =\, $& $\!\!\con{g}$ & if $\exists(h, e_1)\! \in\! \rho\,.\,e_1  \homemb \con{g}$ and $\con{g} \homemb e_1$\hspace{-1em} & (2) \\
$\adr{g}{\R}$ &$\! =\, $& $\!\! [\dr{\overline{f}}{\boxempty}/\overline{y}] \dr{f_g}{\boxempty}$  & if $\exists(h, e_1)\! \in \!\rho\,.\, e_1 \homemb \con{g}$ & (3)\\
 \multicolumn{4}{l}{\hspace{2em} where $(f_g, \overline{f}, \overline{y}) = split(\con{g}, e_1)$ } \\ 
$\adr{g}{\R}$ &$\! =\, $& $\!\! [\dr{\overline{f}}{\boxempty}/\overline{y}] \dr{f_g}{\boxempty}$ &if $\exists e_1 \in e \,.\,e_1 \homemb \con{g}$ & (4a)\\
& & $\!\!\LETREC h = \lambda \overline{x}.e \IN h\,\overline{x}$ & if $h \in \fn(e)$ & (4b) \\
 & & $\!\!e$ & otherwise & (4c) \\
 \multicolumn{5}{l}{\hspace{2em} where $(g, v) \in $ {\it G},} \\
 \multicolumn{5}{l}{\hspace{5.4em}$e = \drp{\con{v}}{\boxempty}$,} \\
 \multicolumn{5}{l}{\hspace{5.4em}$\rho' = \rho \cup (h, \con{g})$,} \\
\multicolumn{5}{l}{\hspace{5.4em}{\em h} fresh, }\\
 \multicolumn{5}{l}{\hspace{5.4em}$\overline{x} = \fv(\con{g})$,} \\
\multicolumn{5}{l}{\hspace{5.4em}$(f_g, \overline{f}, \overline{y}) = split(\con{g}, e_1)$} \\
\end{tabular}
\caption{Driving of applications}
\label{fig:adrivalg}
\end{figure*}

Care needs to be taken to ensure that recursive functions are not
inlined forever. The driving algorithm keeps track of previously
seen function applications in the memoization list $\rho$, which also associates a unique function name to each such expression. Whenever an expression that is equivalent up to renaming of variables to a previous application, a call to the associated function symbol is inserted instead. 
 This is not sufficient to guarantee
termination of the algorithm, but the mechanism is crucial for the
complexity improvements mentioned in Section \ref{sec:examples}.

To ensure termination, we use the homeomorphic embedding relation
$\homemb$ to define a predicate called ``the whistle''. When the
predicate holds for an expression we say that the whistle blows on
that expression. The intuition is that when $e \homemb f$, {\em f}
contains all subexpressions of {\em e}, possibly embedded in other
expressions. For any infinite sequence $e_0, e_1, \ldots$ there must exist an {\em
  i} and a {\em j} such that $i < j$ and $e_i \homemb e_j$. This
condition is sufficient to ensure termination.

In order to define the homeomorphic embedding we need a definition of
uniform terms analogous to the one defined by
\citet{Soerensen:1995:AnAlgorithm}. We slightly adjust their version
to fit our language.

\begin{defi}[Uniform terms]
\label{def:unifterm}
Let {\it s} range over the set $\mathcal{G} \cup K \cup \{\hskwd{caseof},
\hskwd{let}, \hskwd{letrec},$ $\hskwd{primop},\hskwd{lambda},
\hskwd{apply} \}$, and let
$\hskwd{caseof}(\overline{e}),\hskwd{let}(\overline{e}),
\hskwd{letrec}(\overline{v}, e), \hskwd{primop}(\overline{e}),
\hskwd{lambda}(e)$, and $\hskwd{apply}(\overline{e})$ denote a case,
let, recursive let, primitive operation, lambda abstraction or
application for all subexpressions $\overline{e}, e$ and
$\overline{v}$.  The set of terms {\it T} is the smallest set of arity
respecting symbol applications $s(\overline{e})$.
\end{defi}

\begin{defi}[Homeomorphic embedding]
  Define $\homemb$ as the smallest relation on $T$
  satisfying: 
$$
x \homemb y\qquad   n_1 \homemb n_2 \qquad \frac{e \homemb f_i \text{ for
     some i}}{e \homemb s(f_1,\ldots,f_n)}  \qquad
\frac{e_1 \homemb f_1,
   \ldots, e_n \homemb f_n}{s(e_1,\ldots,e_n) \homemb
   s(f_1,\ldots,f_n)}
$$
\end{defi}

Whenever the whistle blows, our transformation splits the input
expression into strictly smaller terms that are driven separately in
the empty context. This might expose new folding opportunities, and
allows the algorithm to remove intermediate structures in
subexpressions.  The design follows the positive supercompilation algorithm 
outlined by \citet{Soerensen:2000:Convergence}, except that we need to
reassemble the transformed subexpressions into a term of the original form instead of pulling them out as let-definitions, in order to preserve strictness. Our
transformation is also more complicated because we perform the program
extraction immediately, rather than constructing a large tree of terms and
extracting the program in a separate pass.

Splitting expressions is rather intricate, and two mechanisms are
needed; the first is the most specific generalization ({\it msg}).

\begin{defi}[Most specific generalization] ~
\begin{enumerate}[$\bullet$]
\item An {\em instance} of a term {\em e} is a term of the form
  $\theta e$ for some substitution $\theta$.
\item A {\em generalization} of two terms {\em e} and {\em f} is a triple
  $(t_g, \theta_1, \theta_2)$, where $\theta_1, \theta_2$ are
  substitutions such that $\theta_1 t_g \equiv e$ and $\theta_2 t_g
  \equiv f$.
\item A {\em most specific generalization} (msg) of two terms {\em e}
  and {\em f} is a generalization $(t_g, \theta_1, \theta_2)$ such
  that for every other generalization $(t_g', \theta_1', \theta_2')$
  of {\em e} and {\em f} it holds that $t_g$ is an instance of $t_g'$.
\end{enumerate}
\end{defi}

For background information and
an algorithm to compute most specific generalizations, see 
\citet{Lassez:1988:Unification}. Figure \ref{fig:hommsg} contains
examples of the homeomorphic embedding and the msg.

\begin{figure*}
\begin{center}
\begin{tabular*}{0.8\textwidth}{@{\extracolsep{\fill}}ccccccc}
e & & f & $t_g$ & $\theta_1$ & $\theta_2$ \\
\hline
\<e\> & $\homemb$ & \<Just e\> & \<x\> & $[e/x]$ & $[Just\, e/x]$ \\
\<Right e\> & $\homemb$ & \<Right (e, e')\> & \<Right x\> & $[e/x]$ & $[(e, e')/x]$ \\
\<fac y\> & $\homemb$ & \<fac (y - 1)\> & \<fac x\> & $[y/x]$ & $[(y-1)/x]$ \\
\end{tabular*}
\end{center}
\caption{Examples of the homeomorphic embedding and the msg}
\label{fig:hommsg}
\end{figure*}

The most specific generalization is not always sufficient to split
expressions. For expressions differing already in their roots, {\it msg} will return just a variable and substitutions equal to the input terms on that variable.  If
this happens we need to split expressions in a different way. We
therefore define our function {\it split} using two alternatives; one that applies when there is a non-trivial most specific generalization, and one that just splits along the spine of the first term in the other case.

\begin{defi}[Split]
For $t \in T$ we define $split(t_1, t_2)$ by: \\
\begin{tabular}{llll}
$split(s(\overline{e}_1), s'(\overline{e}_2))$ & = &  $(t_g, \rng (\theta_1), \dom(\theta_1))$ & if $s = s'$\\
& = & $ (s(\overline{x}), \overline{e}_1, \overline{x})$   & otherwise \\
\end{tabular}\\
with $(t_g, \theta_1, \theta_2) = msg(s(\overline{e}_1), s'(\overline{e}_2))$ and $\overline{x}$ fresh.
\end{defi}

Alternatives 2 and 4a of \<\asdr{}\aedr\> is for upwards
generalization, and alternative 3 is for downwards
generalization. This is exemplified below. All the examples of how
our transformation works in Section \ref{sec:examples} eventually terminate
through a combination of alternative 1 and alternative 4b of \<\asdr{}\aedr\>.

The second alternative of \<\asdr{}\aedr\> in combination with 4a is
useful when transforming function calls that have the same parameter
appearing twice, for example {\em append xs xs} as shown in Figure
\ref{fig:appxsxs}.

\begin{figure}
\begin{tabular}{lp{0.8\linewidth}l}
~ \\
& \<
 \sdr{}append xs xs\edr
\> (*)\\
~ \\
& (By rule 4 of \<\asdr{}\aedr\>, put ($h_0$, \<append xs xs\>) in $\rho$ and transform 
   according to the rules of the algorithm) \\
~ \\
= &
\<\hscase{xs}{[] \to xs \\ (x':xs') \to \sdr{}x'\edr :\sdr{}append xs' xs\edr}
\> \\
~ \\
& (Focus on \<\sdr append xs' xs\edr\> and recall that 
   $\rho$ contains \<append xs xs\> so alternative 2 of \<\asdr
   \aedr\> is triggered and the transformation returns 
   \<append xs' xs\>. This returns all the way up to the start (*)
   and the transformation continues there through alternative 4a)
   \\
~ \\
= & \<
 \sdr{}append xs xs\edr
\> \\
~ \\
& (Generalize the expression with \<append xs' xs\>)  \\
~ \\
= & \<
 [\sdr{}xs\edr/x, \sdr{}xs\edr/y] \sdr{}append x y\edr
\> \\
~ \\
= & \<
 [xs/x, xs/y] \hscase{x}{[] \to y \\ (x':xs') \to \sdr{}x'\edr :\sdr{}append xs' y\edr}
\> \\
~ \\
= & \<
 [xs/x, xs/y] \hscase{x}{[] \to y \\ (x':xs') \to x' :h_0 xs' y}
\> \\
~ \\
= & \<
\LETREC h_0 xs ys = \hscase{xs}{[] \to ys\\(x':xs') \to x':h_0  xs' ys} \> \\
& \<\IN h_0 xs xs\>
\end{tabular}
\caption{Example of upwards generalization}
\label{fig:appxsxs}
\end{figure}

The third alternative is used when terms are ``growing'' in some
sense. An example of {\em reverse} with an accumulating parameter is
shown in Figure \ref{fig:accpar}, assuming the standard definition of reverse.


\begin{figure}
\begin{tabular}{lp{0.8\linewidth}l}
~ \\
& \<
 \sdr{}rev xs []\edr
\> \\
~ \\
& (By rule 4 of \<\asdr{}\aedr\>, put ($h_0$, \<rev xs []\>) in $\rho$ and transform the program
   according to the rules of the algorithm) \\
~ \\
 &
\<\hscase{xs}{[] \to [] \\ (x':xs') \to \sdr{}rev xs' (x':[])\edr}
\> \\
~ \\
& (Focus on the second branch and recall that $\rho$ contains \<rev xs
  []\> so alternative 3 of \<\asdr \aedr\> is triggered and the
  expression is generalized) \\ 
~ \\
= &
\< \sdr{}rev xs' (x':[])\edr
\> \\
~ \\
 & (Generalize the expression with \<rev xs []\>)  \\
~ \\
= &
\< [\sdr{}(x':[])\edr/zs]\sdr{}rev xs' zs\edr
\> \\
~ \\
= & \< [(x':[])/zs]\sdr{}rev xs' zs\edr 
\> \\
~ \\
 & (Put  ($h_1$, \<rev xs' zs\>) in $\rho$ and transform 
   according to the rules of the algorithm) \\
~ \\
= & \<[(x':[])/zs]\LETREC h_1 xs ys = \hscase{xs}{[] \to ys\\ (x':xs') \to h_1 xs' (x':ys)}
\> \\
& \hspace{4.8em}\< \IN h_1 xs' (x':[])
\> \\
~ \\
= & \< \LETREC h_1 xs ys = \hscase{xs}{[] \to ys\\ (x':xs') \to h_1 xs' (x':ys)}
\> \\
& \< \IN h_1 xs' (x':[])
\> \\
~ \\
& (Putting the two parts together) \\
~ \\
& \< \hscase{xs}{[] \to [] \\ (x':xs') \to \LETREC h_1 xs ys = \hscase{xs}{[] \to\!\! ys\\ (x':xs') \to\!\! h_1 xs' (x':ys)} \\
\qquad \quad \quad \quad \quad \quad \quad \quad \quad \quad \quad \quad \quad \quad \quad \quad \quad \quad \quad \quad \quad \quad \quad \quad \quad \quad \IN h_1 xs' (x':[]) }
\> \\
\end{tabular}
\caption{Example of downwards generalization}
\label{fig:accpar}
\end{figure}

\section{Correctness}
\label{sec:correctness}

The problem with using previous deforestation and supercompilation
algorithms in a call-by-value context is that they might change
the termination properties of programs. In this section we prove that our supercompiler
both terminates itself, and preserves program termination behavior for all input. 

\subsection{Termination}

In order to prove that the algorithm terminates we show that each
recursive application of \<\sdr{}\edr\> in the right-hand sides of
Figure \ref{fig:drivalg} and \ref{fig:adrivalg} has a strictly
smaller weight than the left-hand side. 

The weight of an expression
is one plus the sum of the weight of its subexpressions, where
variables, primitive numbers and function names have weight two. The
weight of a fresh variable not in the initial input is one.

\begin{defi}
\label{def:expweight}
The weight of a variable {\it x} in the initial input, a primitive number {\it n},
and a function name {\it g} is 2. The weight of a fresh variable not in the
initial input is 1. The weight of any composite expression ($n
\geq 1$) is $|s(e_1,\ldots, e_n)| = 1 + \sum_{i = 1}^n |e_i| $.
\end{defi}

\begin{defi}
  Let {\it S} be a set with a relation $\leq$. Then $(S, \leq)$ is a
  quasi-order if $\leq$ is reflexive and transitive.
\end{defi}

\begin{defi}
  Let $(S, \leq)$ be a quasi-order. $(S, \leq)$ is a well-quasi-order
  if, for every infinite sequence $s_0,s_1, \ldots \in S$, there exist
  $i < j$ such that $s_i \leq s_j$
\end{defi}

The following lemma tells us that the set of finite sequences over a
well-quasi-ordered set is well-quasi-ordered, with one proof by
\citet{Nash-Williams:1963:On}:

\begin{lem}[Higman's lemma]
\label{lem:higman}
If a set $S$ is well-quasi-ordered, then the set $S^{*}$ of finite
sequences over $S$ is well-quasi-ordered. 
\end{lem}

The weight of the entire transformation is a triple that contains the
maximum length of the memoization list $\rho$ denoted by {\em N}, the
weight of the term being transformed and the weight of the current
term in focus. That such an {\em N} exists follows from Kruskal's Tree
Theorem \citep{Dershowitz:1987:Termination} and the homeomorphic
embedding relation being a well-quasi-order.

\begin{thm}[Kruskal's Tree Theorem]
\label{thm:krusktree}
  If {\it S} is a finite set of function symbols, then any infinite sequence
  $t_1, t_2, \ldots$ of terms from the set {\it S} contains
  two terms $t_i$ and $t_j$ with $i < j$ such that $t_i \homemb t_j$.
\end{thm}

\begin{proof}[Proof (Similar to \citet{Dershowitz:1987:Termination})]
  Collapse all integers to a single 0-ary constructor, and all
  variables to a different 0-ary constructor.

Suppose the theorem were false. Let the infinite sequence
$\overline{t} = t_1, t_2, \ldots$ of terms be a minimal
counterexample, measured by the size of the $t_i$. By the minimality
hypothesis, the set of proper subterms of the $t_i$ must be
well-quasi-ordered, or else there would be a smaller counterexample
$t_1, t_2, \ldots, t_{l-1}, s_1, s_2, \ldots$, for some {\it l} such that $s_1$ is a subterm of $t_1$ and all $s_2,\ldots$ are subterms of one of $t_l, t_{l+1},\ldots$. (None of $t_1, t_2, \ldots, t_{l-1}$
can embed any of $s_1, s_2, \ldots$, since that would mean that $t_i$
also is embedded in some $t_j, i < l \leq j$).

Since the set {\it S} of function symbols is well-quasi-ordered by $\geq$,
there must exist an infinite subsequence $\overline{r}$ of
$\overline{t}$, the root (outermost) symbols of which constitute a
quasi-ascending chain under $\leq$. (Any infinite sequence
of elements of a well-quasi-ordered set must contain an infinite chain
of quasi-ascending elements). Since the set of proper subterms is
well-quasi-ordered, it follows by Lemma \ref{lem:higman} that the
set of finite sequences consisting of the immediate subterms of the
elements in $\overline{r}$ is also well-quasi-ordered. But then there would have to
be an embedding in $\overline{t}$ itself, in which case it would not
be a counterexample.
\end{proof}

We will show that each step of the driving algorithm will reduce the
weight of what is being transformed. The constant {\em N} in the
weight is the maximum length of the sequence of terms that are not
related to each other by the homeomorphic embedding. 

\begin{cor}
Any infinite sequence $t_1,t_2, \ldots \in T^*$ contains two terms $t_i$ and $t_j$ with $i < j$ such that $t_i \homemb t_j$.
\end{cor}

\begin{cor}
\label{cor:maxn}
There is a maximum {\it N} such that $t_1,t_2, \ldots, t_N \in T^*$ contains no terms $t_i$ and $t_j$ with $i < j$ and $t_i \homemb t_j$.
\end{cor}

\noindent We define the weight of driving a term as:
\begin{defi}
\label{def:drvweight}
The weight of a call to the driving algorithm is $|\dr{e}{\R}| = (N -
|\rho|, |\con{e}|, |e|)$
\end{defi}

Tuples must be ordered for us to tell whether the weight of a term
actually decreases from driving it. We use the standard 
lexical order between tuples.

\begin{defi}
\label{def:order}
The order between two tuples $(n_1,n_2, n_3)$ and $(m_1, m_2, m_3)$ is:
$$
\begin{array}{ll}
  (n_1, n_2, n_3) < (m_1, m_2, m_3) & \text{if $n_1 < m_1$}\\
  (n_1, n_2, n_3) < (m_1, m_2, m_3) & \text{if $n_1 = m_1$ and $n_2 < m_2$}\\
  (n_1, n_2, n_3) < (m_1, m_2, m_3) & \text{if $n_1 = m_1$, $n_2 = m_2$ and $n_3 < m_3$}\\
\end{array}
$$
\end{defi}

We also need to show that the memoization list $\rho$ only contains
elements that were in the initial input program:
\begin{lem}
\label{lem:extrho}
  The second component of the memoization list $\rho$, can only
  contain terms from the set {\it T}.
\end{lem}

\begin{proof}
  Integers and fresh variables are equal, up to $\homemb$, to the
  already existing integers and variables. Our only concern are the
  rules that introduce new terms that are not in {\it T}. The new function
  names $h$ are the only new terms introduced by the algorithm. By
  inspection of the rules it is clear that only rule R3 introduces
  such new terms. Inspection of the RHS of $\adr{}{}$:

  \begin{enumerate}[\bf\phantom01:]
    \item No recursive application in the RHS.
    \item No recursive application in the RHS.
    \item No new terms are created and the memoization list
        $\rho$ is not extended.
   \item[\bf 4a:] No new terms are created and the memoization list
        $\rho$ is not extended.
    \item[\bf 4b:] The newly created term $h\,\overline{x}$ is
      kept outside of the recursive call of the driving algorithm. The
      memoization list $\rho$, is extended with terms from {\it T}.
    \item[\bf 4c:] No new terms are created,
      and the memoization list $\rho$, is extended with terms from {\it T}.
  \end{enumerate}
\end{proof}

With these
definitions in place, we can formulate a lemma that claims the weight is
decreasing in each step of our transformation.

\begin{lem}
\label{lem:weightdec}
  For each rule $\dr{e}{\R} = e_1$ in Figure
  \ref{fig:drivalg} and Figure \ref{fig:adrivalg} and each recursive
  application $\D\llbracket e'\rrbracket_{\R',{\it G},\rho'}$ in $e_1$, 
  $|\D\llbracket e'\rrbracket_{\R',{\it G},\rho'}| < |\dr{e}{\R}|$
\end{lem}

\begin{lem}[Totality]
\label{lem:totality}
For all expressions $\con{e}$, $\dr{e}{\R}$ is matched by a unique
rule in Figure \ref{fig:drivalg}.
\end{lem}

\begin{thm}[Termination]
  The driving algorithm \<\sdr{}\edr\> terminates for all
   inputs.
\end{thm}

\begin{proof}
  The weight of the transformation is defined because of Kruskal's
  Tree Theorem and the fact that the homeomorphic embedding is a
  well-quasi-order. Lemma \ref{lem:extrho} guarantees that the
  memoization list $\rho$ only contains terms from the initial
  input. By Lemma \ref{lem:weightdec} the weight of the transformation
  decreases for each step and by Lemma \ref{lem:totality} we know that
  each recursive application will match a rule.

  Since $<$ is well-founded over triples of natural numbers the system
  will eventually terminate.
\end{proof}

\subsection{Total Correctness}

The problem with previous deforestation and supercompilation
algorithms in a call-by-value context is that they might change
termination properties of programs. We prove that our supercompiler
does not change what the program computes, nor does it alter whether a program
terminates or not.

\citet{Sands:1996:Proving} shows how a transformation can change the
semantics in rather subtle ways~-- consider the function
\begin{haskell*}
f x = x + 42
\end{haskell*}

It is clear that \<f 0 \cong 42\> (where $\cong$ is semantic
equivalence with respect to the current definition). Using this
equality and replacing 42 in the function body with \<f 0\> yields:

\begin{haskell*}
f x = x + f 0
\end{haskell*}

This function will compute something entirely different than the
original definition of \<f\>. We need some tools to ensure that the
meaning of the original program is preserved and we therefore
introduce the standard notions of operational approximation and
equivalence. A general context {\em C} which is an expression with zero or
more holes in the place of some subexpressions is used, and we say that an
expression $C[e]$ is {\em closed} if there are no free variables in it.

\begin{defi}[Operational Approximation and Equivalence]
~ 
\begin{enumerate}[$\bullet$]
\item  $e$ operationally approximates $e'$, $e \opapp e'$, if for all contexts
  $C$ such that $C[e]$, $C[e']$ are closed, if evaluation of $C[e]$ terminates
  then so does evaluation of $C[e']$.
\item  $e$ is operationally equivalent to $e'$, $e \opeq e'$, if $e \opapp
  e'$ and $e' \opapp e$
\end{enumerate}
\end{defi}

The correctness of deforestation in a call-by-name setting has
previously been shown by  \citet{Sands:1996:Proving} using his
improvement theory. We use Sands's definitions for improvement and
strong improvement:

\begin{defi}[Improvement, Strong Improvement]
~ 
\begin{enumerate}[$\bullet$]
\item  $e$ is improved by $e'$, $e \costrighteq e'$, if for all contexts $C$ such
  that $C[e]$, $C[e']$ are closed, if computation of $C[e]$ terminates using
  {\it n} function calls, then computation of $C[e']$ also terminates, and
  uses no more than {\it n} function calls.
\item  $e$ is strongly improved by $e'$, $e \costrighteq_s e'$, iff $e
  \costrighteq e'$ and $e \opeq e'$.
\end{enumerate}
\end{defi}

Note that improvement, $\costrighteq$, is
not the same as the homeomorphic embedding, $\homemb$, defined
previously.

We use $e \eval^{k} v$ to denote that {\em e} evaluates to {\em v}
using {\em k} function calls (and any other reduction rule as many
times as it needs) and $e' \eval^{\leq k} v'$ to denote that {\em e'}
evaluates to {\em v'} with at most {\em k} function calls and using any
other reduction rule as many times as needed.

To state the Improvement Theorem  we view a transformation as the
introduction of some new functions from a given set of definitions. We
let $\{g_i\}_{i \in I}$ be a set of functions indexed by
some set {\em I}, where each function has a fixed arity $\alpha_i$ and are given by some definitions
\begin{equation*}
\{g_i = \lambda x_1\ldots x_{\alpha_i} . e_i\}_{i \in I}
\end{equation*}
and let $\{e_i'\}_{i \in I}$ be a set of expressions such that for
each $i \in I, \fv(e_i') \subseteq \{x_1\ldots x_{\alpha_i}\}$. 
The following results relate to the transformation of the functions
$g_i$ using the expressions $e_i'$: let $\{h_i\}_{i \in I}$ be a set
of new functions given by the definitions
\begin{equation*}
\{h_i = [\overline{h}/\overline{g}]\lambda x_1\ldots x_{\alpha_i} . e_i'\}_{i \in I}
\end{equation*}

\begin{thm}[Sands Improvement theorem]
\label{thm:sandsimp}
 If $g = e$ and $e \costrighteq C[g]$ then $g \costrighteq h$ where $h
 = C[h]$. 
\end{thm}

\begin{thm}[Cost-equivalence theorem]
\label{thm:sandscost}
If $e_i \costeq e_i'$ for all $i \in I$, then $g_i \costeq h_i$, $i
\in I$.
\end{thm}

We need a standard partial correctness result
\citep{Sands:1996:Proving} associated with unfold-fold
transformations

\begin{thm}[Partial Correctness]
\label{thm:partcorr}
If $e_i \cong e_i'$ for all $i \in I$ then $h_i \opapp g_i$, $i \in
I$. 
\end{thm}
\noindent which we combine with Theorem \ref{thm:sandsimp} 
to get total correctness for a
transformation:

\begin{cor}
\label{cor:simp}
If we have $e_i \costrighteq_s e_i'$ for all $i \in I$, then $g_i
\costrighteq_s h_i$, $i \in I$.
\end{cor}

Improvement theory in a call-by-value setting requires Sands
operational metatheory for functional languages
\citep{Sands:1997:From}, where the improvement theory is a simple
corollary over the well-founded resource structure $\langle
\mathbb{N}, 0, +, \geq \rangle$. For simplicity of presentation we
instantiate Sands's theorems to our language.  We use $\equiv$ to
denote expressions equal up to renaming of bound variables and borrow
a set of improvement laws that will be useful for our proof:

\begin{lem}[\citet{Sands:1996:Total}]
\label{lem:implaws}
Improvement laws
\begin{enumerate}[\em(1)]
\item If $e \costrighteq e'$ then $C[e] \costrighteq C[e']$.
\item If $e \equiv e'$ then $e \costrighteq e'$.
\item If $e \costrighteq e'$ and $e' \costrighteq e''$ then $e \costrighteq e''$
\item If $e \eval e'$ then $e \costrighteq e'$.
\item If $e \costrighteq e'$ then $e \opapp e'$.
\end{enumerate}
\end{lem}

It is sometimes convenient to show that two expressions are related by
showing that what they evaluate to is related. 

\begin{lem}[\citet{Sands:1996:Proving}]
\label{lem:steps}
If $e_1 \eval^r e_1'$ and $e_2 \eval^r e_2'$ then 
($e_1' \costeq e_2' \Leftrightarrow e_1 \costeq e_2$).
\end{lem}

We need to show strong improvement in order to prove total correctness. Since strong
improvement is improvement in one direction and operational
approximation in the other direction, a set of approximation laws that
correspond to the improvement laws in Lemma \ref{lem:implaws} is
necessary.

\begin{lem}
\label{lem:applaws}
Approximation laws
\begin{enumerate}[\em(1)]
\item If $e \bopapp e'$ then $C[e] \bopapp C[e']$.
\item If $e \equiv e'$ then $e \bopapp e'$.
\item If $e \bopapp e'$ and $e' \bopapp e''$ then $e \bopapp e''$
\item If $e \eval e'$ then $e \bopapp e'$.
\end{enumerate}
\end{lem}

Combining Lemma \ref{lem:implaws} and Lemma \ref{lem:applaws} gives us the
final tools we need to prove strong improvement: 

\begin{lem}
\label{lem:simplaws}
Strong Improvement laws
\begin{enumerate}[\em(1)]
\item If $e \costrighteq_s e'$ then $C[e] \costrighteq_s C[e']$.
\item If $e \equiv e'$ then $e \costrighteq_s e'$.
\item If $e \costrighteq_s e'$ and $e' \costrighteq_s e''$ then $e \costrighteq_s e''$
\item If $e \eval e'$ then $e \costrighteq_s e'$.
\end{enumerate}
\end{lem}

If two expressions are improvements of each other, they are considered
cost equivalent. Cost equivalence also implies strong improvement,
which will be useful in many parts of our proof of total correctness
for our supercompiler.

\begin{defi}[Cost equivalence]
  The expressions $e$ and $e'$ are cost equivalent, $e \costeq e'$ iff $e \costrighteq e'$
  and $e' \costrighteq e$ 
\end{defi}

A local form of the improvement theorem which deals with local
expression-level recursion expressed with a fixed-point combinator or
with a letrec definition is necessary. This is analogous to the work
by \citet{Sands:1996:Proving}, with slight modifications for
call-by-value.

We need to relate local recursion expressed using {\it fix} and the
recursive definitions which the improvement theorem is defined for.
This is solved by a technical lemma that relates the cost of terms on
a certain form to their recursive counterparts.

\begin{thm}
\label{prop:fix}
For all expressions $e$, if $\lambda g.e$ is closed, then $fix\,
(\lambda g.e) \costeq h$, where {\it h} is a new function defined by
$h = [\lambda n.h\,n/g]e$.
\end{thm}

\begin{proof}[Proof (Similar to \citet{Sands:1996:Proving})]
  Define a helper function $h^{-} = [\lambda n.fix\,(\lambda
  g.e)\,n/g]e$. Since $fix\,(\lambda g.e) \eval^1 $ $(\lambda
  f.f\,(\lambda n.fix\, f\,n))\, (\lambda g.e) \eval (\lambda
  g.e)\,(\lambda n.fix\, (\lambda g.e)\,n) \eval [\lambda n.fix\,
  (\lambda g.e)\,n/g]e$ and $h^{-} \eval^1 [\lambda n.fix\,(\lambda
  g.e)\,n/g]e$ it follows by Lemma \ref{lem:steps} that $fix\,
  (\lambda g.e) \costeq h^{-}$. Since cost equivalence is a
  congruence relation we have that $[\lambda n.h^{-}\,n/g]e \costeq
  [\lambda n.fix\,(\lambda g.e)\,n/g]e$, and so by Theorem
  \ref{thm:sandscost}, we have a cost-equivalent transformation from
  $h^{-}$ to $h$, where $h = [h/h^{-}][\lambda n.h^{-}\,n/g]e =
  [\lambda n.h\,n/g]e$.
\end{proof}

\noindent We state some simple properties that will be useful for proving our local
improvement theorem

\begin{thm}
\label{prop:letreceq} Consequences of the letrec definition
\begin{enumerate}[\bf i):]
\item $\LETREC h = \lambda \overline{x}.e \IN e' \costeq [\lambda n.fix\,(\lambda h.\lambda \overline{x}.e)\,n/h]e'$
\item $\LETREC h = \lambda \overline{x}.e \IN h \costeq \lambda n.fix\,(\lambda h.\lambda \overline{x}.e)\,n$
\item $\LETREC h = \lambda \overline{x}.e \IN e' \costeq [\LETREC h =
\lambda \overline{x}.e \IN h/h]e'$
\end{enumerate}
\end{thm}

\begin{proof}
  For i), expand the definition of letrec in the LHS, $(\lambda
  h.e')\,(\lambda n.fix\, (\lambda h.\lambda \overline{x}.e)\,n)$ and
  evaluate it one step to  $[\lambda n.fix\, (\lambda h.\lambda
  \overline{x}.e)\,n/h]e'$. This is syntactically equivalent to the
  RHS, hence cost equivalent. For ii), set $e' = h$ and perform the
  substitution from i). For iii), use the RHS of ii) in the
  substitution and notice it is equivalent to i).
\end{proof}

\noindent This allows us to state the local version of the improvement theorem:

\begin{thm}[Local improvement theorem]
\label{thm:letrec-simp}
If variables {\it h} and $\overline{x}$ include all the free variables of
both $e_0$ and $e_1$, then if
\begin{equation*}
\LETREC h = \lambda \overline{x}.e_0 \IN e_0 \costrighteq_s \LETREC h
= \lambda \overline{x}.e_0 \IN e_1
\end{equation*}
then for all expressions e
\begin{equation*}
\LETREC h = \lambda \overline{x}.e_0 \IN e \costrighteq_s \LETREC h
= \lambda \overline{x}.e_1 \IN e
\end{equation*}
\end{thm}

\begin{proof}
  Define a new function $g = [\lambda n.g\,n/h]\lambda
  \overline{x}.e_0$. By Proposition \ref{prop:fix} $g \costeq fix\,
  (\lambda h.\lambda \overline{x}.e_0)$. Use this, the congruence
  properties, and the properties listed in Proposition
  \ref{prop:letreceq} to transform the premise of the theorem:
\begin{align*}
  \LETREC h = \lambda \overline{x}.e_0 \IN e_0 & \costrighteq_s \LETREC h = \lambda \overline{x}.e_0 \IN e_1 \\
  [\lambda n.fix\,(\lambda h.\lambda \overline{x}.e_0)\,n/h]e_0  &\costrighteq_s [\lambda n.fix\,(\lambda h.\lambda \overline{x}.e_0)\,n/h]e_1 \\
  \lambda \overline{x}.[\lambda n.fix\,(\lambda h.\lambda \overline{x}.e_0)\,n/h]e_0  & \costrighteq_s  \lambda \overline{x}.[\lambda n.fix\,(\lambda h.\lambda \overline{x}.e_0)\,n/h]e_1 \\
  [\lambda n.fix\,(\lambda h.\lambda \overline{x}.e_0)\,n/h]\lambda \overline{x}.e_0  &\costrighteq_s [\lambda n.fix\,(\lambda h.\lambda \overline{x}.e_0)\,n/h]\lambda \overline{x}.e_1 \\
  [\lambda n.g\,n/h]\lambda \overline{x}.e_0  &\costrighteq_s [\lambda n.g\,n/h]\lambda \overline{x}.e_1 \\
\end{align*}

So by Corollary \ref{cor:simp}, $g \costrighteq_s g'$ where $g' =
[g'/g][\lambda n.g\,n/h]\lambda \overline{x}.e_1 = [\lambda
n.g\,n/h]\lambda \overline{x}.e_1$. Hence by Proposition
\ref{prop:fix}, $g' \costeq fix\, (\lambda h.\lambda
\overline{x}.e_1)$. Adding it all together yields $fix\,(\lambda
h.\lambda \overline{x}.e_0)$ $\costeq g \costrighteq_s g' \costeq
fix\, (\lambda h.\lambda \overline{x}.e_1)$. From the transitivity and
congruence properties of improvement we can deduce that $\lambda
n.fix\,(\lambda h.\lambda \overline{x}.e_0) \costrighteq_s \lambda
n.fix\,(\lambda h.\lambda \overline{x}.e_1)$. By Proposition
\ref{prop:letreceq} we get $\LETREC h = \lambda \overline{x}.e_0 \IN h $ $
\costrighteq_s \LETREC h = \lambda \overline{x}.e_1 \IN h$, which can
be further expanded by congruency properties of improvement to
$[\LETREC h = \lambda \overline{x}.e_0 \IN h/h]e \costrighteq_s
[\LETREC h = \lambda \overline{x}.e_1\IN h/h]e$. Using Proposition
\ref{prop:letreceq} one more time yields $\LETREC h = \lambda
\overline{x}.e_0 \IN e \costrighteq_s \LETREC h = \lambda
\overline{x}.e_1 \IN e$, which proves our theorem.
\end{proof}

\noindent This allows us to state the total correctness theorem for our transformation:

\begin{thm}[Total Correctness]
\label{prop:totcorr}
Let $\con{e}$ be an expression, {\it G} a recursive map, and $\rho$ an environment such that
\begin{enumerate}[$\bullet$]
\item the range of $\rho$ contains only closed expressions, and
\item $\fv(\con{e}) \cap dom(\rho) = \emptyset$, and
\end{enumerate}
then $\con{e} \costrighteq_s \rho(\dr{e}{\R})$.
\end{thm}

\noindent The proof is in Appendix \ref{sec:startbev} to Appendix 
\ref{sec:slutbev}.

\section{Benchmarks}
\label{sec:benchmarks}

In this section we provide measurements on a set of common examples from the
literature on deforestation and perform a detailed analysis for each
example. We show that our positive supercompiler removes intermediate
structures and can improve the performance by an order of magnitude
for certain benchmarks. The supercompiler was implemented as a pass in
the Timber compiler \citep{Timber}. Timber is a pure
functional call-by-value language which is very close to the language
we describe in Section \ref{sec:lang}, and for the scope of this article
it can be thought of as a strict variant of Haskell. We have left out the full details of the
instrumentation of the run-time system but it is available in a
separate report \citep{Jonsson:2008:Positive}.

All measurements were performed on an idle machine running in an
{\it xterm} terminal environment. Each test was run 10 consecutive times and the best result was
selected because the programs are deterministic and the best result must appear under the minimum of other activity. The number of
allocations and the total allocation sizes remained constant over all runs.

Raw data for the time and size measurements before and after supercompilation are shown in Table
\ref{fig:tsmeasurements}, and allocation measures in Table
\ref{fig:ameasurements}. Compilation times are shown in Table
\ref{fig:cmeasurements}. The time column is the number of clock ticks
obtained from the {\it RDTSC} instruction available on Intel/AMD processors, and the
binary size is in bytes. The total number of allocations and the total
memory size in bytes allocated by the program are displayed in their
respective column. The compilation times are measured in seconds and
times from left to right are for producing an object file,
producing an executable binary, and the corresponding operations with
supercompilation turned on.

Binary sizes are slightly increased by the supercompiler, but all
run-times are faster. The main reason for the performance improvement
is the removal of intermediate structures, reducing the number of
memory allocations. Compilation times are increased by 10-15\% when 
enabling the supercompiler. 

The supercompiled results on these particular benchmarks are identical
to the results reported in previous work for call-by-name languages by
\citet{Wadler:1990:Deforestation} and
\citet{Soerensen:1996:Positive}. We do not provide any execution-time comparisons with these, though, since for identical intermediate representations after supercompilation, such measurements would only illustrate differences caused by back-end implementation techniques.

The work on Supero by \citet{Mitchell:2008:ASupercompiler} shows that
there remain open problems when supercompiling large Haskell
programs. These problems are mainly related to speed, both of the
compiler and of the transformed program. When profiling Supero, Mitchell and Runciman found that a majority of the time was spent on their homeomorphic
embedding test. Our transformation performs the corresponding test on a smaller part of
the abstract syntax tree, so there is reason to believe that this will result in less time spent
on testing homeomorphic embedding even on large programs for our transformation. The complexity of the homeomorphic
embedding relation has been investigated  by
\citet{Narendran:1987:On}, and they give an algorithm of complexity
$O(size(e) \times size(f))$ for deciding whether $e \costlefteq f$.  We
expect essentially the same problems that Mitchell and Runciman observed to appear
in a call-by-value context as well, and intend to investigate them now
that we have a theoretical foundation for our transformation.

\begin{table*}
\begin{tabular}{lr@{\hspace{2em}}r@{\hspace{2em}}r@{\hspace{2em}}r}
\hline
   & \multicolumn{2}{c}{Time} &  \multicolumn{2}{c}{Binary size} \\
  Benchmark & Before & After  & Before & After  \\
\hline
Double Append &  105,844,704 & 89,820,912 & 89,484 & 90,800 \\
Factorial & 21,552 & 21,024 & 88,968 & 88,968 \\
Flip a Tree & 2,131,188 & 237,168  & 95,452 & 104,704 \\ 
Sum of Squares of a Tree & 276,102,012 & 28,737,648 & 95,452 & 104,912 \\
Kort's Raytracer & 12,050,880 & 7,969,224 &  91,968 & 91,460 \\
\end{tabular}
\caption{Time and size measurements}
\label{fig:tsmeasurements}
\end{table*}
\begin{table*}
\begin{tabular}{lr@{\hspace{2em}}r@{\hspace{2em}}r@{\hspace{2em}}r}
\hline
   & \multicolumn{2}{c}{Allocations} & \multicolumn{2}{c}{Alloc Size}  \\
  Benchmark &  Before & After  & Before & After \\
\hline
Double Append &  270,035 & 180,032 & 2,160,280 & 1,440,256  \\
Factorial & 9 & 9 & 68 & 68  \\
Flip a Tree & 20,504 & 57 & 180,480 & 620\\
Sum of Squares of a Tree &  4,194,338 & 91 & 29,360,496 & 908\\
Kort's Raytracer &  60,021 & 17 & 320,144 & 124 \\
\end{tabular}
\caption{Allocation measurements}
\label{fig:ameasurements}
\end{table*}

\begin{table*}
\begin{tabular}{lrrrr}
\hline
   & \multicolumn{2}{c}{Not Supercompiled} & \multicolumn{2}{c}{Supercompiled}\\
  Benchmark &  -c & --make & -c -S  & --make -S \\
\hline
Double Append &  0.183 & 0.300 & 0.202 & 0.319 \\
Factorial & 0.095 & 0.213 &  0.097 & 0.216 \\
Flip a Tree & 0.211 & 0.223 & 0.230 & 0.347 \\
Sum of Squares of a Tree & 0.214 &  0.332 & 0.234 & 0.349\\
Kort's Raytracer &  0.239 & 0.359 & 0.278 & 0.399 \\
\end{tabular}
\caption{Compilation times}
\label{fig:cmeasurements}
\end{table*}

\subsection{Double Append}
As previously seen, supercompiling the appending of three lists saves one traversal over the
first list. This is an example by \citet{Wadler:1990:Deforestation},
and the intermediate structure is fused away by our
supercompiler. The program is:

\begin{haskell}
append xs ys &=& \hscase{xs}{[] \to ys\\ (x':xs') \to x':(append xs' ys)} \\
~ \\
main xs ys zs &=&\,\, append (append xs ys) zs
\end{haskell}

\noindent Supercompiling this program gives the same result that we obtained manually 
in Section \ref{sec:examples}:

\begin{haskell}
h_1 xs_1 ys_1 zs_1 &=& \hscase{xs_1}{
              [] \to \hscase{ys_1}{
                    [] \to zs_1 \\
                    (y_1':ys_1') \to  y_1':(h_2 ys_1' zs_1)} \\
              (x_1':xs_1') \to  x_1':(h_1 xs_1' ys_1 zs_1)} \\
h_2 xs_2 ys_2 &=& \hscase{xs_2}{
              [] \to ys_2\\
              (x_2':xs_2')  \to   x_2':(h_2 xs_2' ys_2)} \\
~ \\
main xs ys zs &=& h_1 xs ys zs
\end{haskell}

In this measurement, three strings of 9000 characters each were appended to
each other into a 27 000 character string.  As can be seen in Table \ref{fig:ameasurements}, the number of 
allocations goes down as one iteration over the first string is
avoided. The binary size increases 1316 bytes, on a binary of roughly
90k.

\subsection{Factorial}

There are no intermediate lists created in a standard implementation of a factorial function,
so any performance improvements must come from inlining or static reductions.

\begin{haskell}
fac 0 & = & 1 \\
fac n & = & n*fac (n-1) \\
~ \\
main &=& show (fac 3) \\
\end{haskell}

\noindent The program is transformed to:

\begin{haskell}
h 0 & = & 1 \\
h n & = & n * h (n - 1)\\
~ \\
main &=& show (3*h 2) \\
\end{haskell}

One recursion and a couple of reductions are eliminated, thereby
slightly reducing the run-time. The allocations remain the same and the
final binary size remains unchanged.

\subsection{Flip a Tree}

Flipping a tree is another example by
\citet{Wadler:1990:Deforestation}, and just like Wadler we perform a
double flip (thus restoring the original tree) before printing the
total sum of all leaves. 

\begin{haskell}
data Tree a = Leaf a | Branch (Tree a) (Tree a) \\
~ \\
sumtr (Leaf a) = a \\
sumtr (Branch l r) = sumtr l + sumtr r \\
~ \\
flip (Leaf x) = Leaf x \\
flip (Branch l r) = Branch (flip r) (flip l) \\
~ \\
main xs = let ys = (flip (flip xs)) in show (sumtr ys) \\
\end{haskell}

\noindent This is transformed into:

\begin{haskell}
h t = \hscase{t}{Leaf d \to d \\
              Branch l r \to (h l) + (h r)} \\
~ \\
main xs = show (\hscase{xs}{
                                 Leaf d \to d \\
                                 Branch l r \to  (h l) + (h r) )}
\end{haskell}

A binary tree of depth 12 was used in the measurement. The function
\<h\> is isomorphic to \<sumtr\> in the input program, and the double
flip has been eliminated. Both the total number of allocations and the
total size of allocations is reduced. The run-time is reduced by an
order of magnitude. The binary size increases by about 10\%, though.

\subsection{Sum of Squares of a Tree}

Computing the sum of the squares of the data members of a tree is the
final example by \citet{Wadler:1990:Deforestation}. 

\begin{haskell}
data Tree a = Leaf a | Branch (Tree a) (Tree a) \\
~ \\
square :: Int \to Int \\
square x = x*x \\
~ \\
sumtr (Leaf x) = x \\
sumtr (Branch l r) = sumtr l + sumtr r \\
~ \\
squaretr (Leaf x) = Leaf (square x) \\
squaretr (Branch l r) = Branch (squaretr l) (squaretr r) \\
~ \\
main xs = show (sumtr (squaretr xs)) 
\end{haskell}

\noindent This is transformed to:

\begin{haskell}
h t = \hscase{t}{
              Leaf d \to d * d \\
              Branch l r \to (h l) + (h r) } \\
~ \\
main xs = show (\hscase{xs}{
                            Leaf d \to d * d \\
                            Branch l r \to (h l) + (h r) } \\
\end{haskell}

Almost all
allocations are removed by our supercompiler, but the binary size is
increased by nearly 10\%. The run-time is 
improved by an order of magnitude.

\subsection{Kort's Raytracer}

The inner loop of a raytracer \citep{Kort:1996:Deforestation} written
in Haskell is extracted and transformed. 

\begin{haskell}
zipWith f (x:xs) (y:ys) = (f x y):zipWith f xs ys \\
zipWith \_ \_ \_ = [] \\
~ \\
sum :: [Int] \to Int \\
sum [] = 0 \\
sum (x:xs) = x + sum xs \\
~ \\
main xs ys = sum (zipWith (*) xs ys)
\end{haskell}

\noindent The transformed result is:

\begin{haskell}
h xs ys = \hscase{xs}{(x':xs') \to
                      \hscase{ys}{(y':ys') \to (x' * y') +
                        (h xs' ys') \\
                        \_ \to 0} \\
              \_ \to  0} \\
~ \\
main xs ys = h xs ys \\
\end{haskell}

The total run-time, the number of
allocations, the total size of allocations and the binary size all
decrease.

\section{Related Work}

There is much literature concerning algorithms that remove
intermediate structures in functional programs. However, most of these works are in the the context of call-by-name or call-by-need languages, which makes the task of supercompilation a
different, yet difficult, problem. We therefore start our survey of
related work with one call-by-value transformation and then look at
the related transformations from a call-by-name or call-by-need
perspective.

\subsection{Lightweight Fusion}

Ohori's and Sasano's Lightweight Fusion \citep{Ohori:2007:Lightweight} works by promoting functions through the fix-point operator and guarantees termination by
limiting inlining to at most once per function.  They implement their
transformation in a compiler for a variant of Standard ML and
present some benchmarks. The algorithm is proven correct for a
call-by-name language.  It is explicitly mentioned that their goal is
to extend the transformation to work for an impure call-by-value
functional language.

Comparing lightweight fusion to our positive supercompiler is somewhat
difficult, the algorithms themselves are not very similar. Comparing
results of the algorithms is more straightforward -- the restriction
to only inline functions once makes lightweight fusion unable to
handle successive applications of the same function or mutually
recursive functions, something the positive supercompiler handles
gracefully.

Despite the early stage of their work, Ohori and Sasano are proposing
an interesting approach that appears quite powerful.

\subsection{Deforestation}

Deforestation was pioneered by \citet{Wadler:1990:Deforestation} for a
first order language more than fifteen years ago. The function macros
supported by the initial deforestation algorithm were not capable of
fully emulating higher-order functions.

\citet{Marlow:1992:Deforestation} addressed the first-order
restriction in a subsequent article
\citep{Marlow:1992:Deforestation}. This work was refined in Marlow's
\citeyearpar{Marlow:1995:Deforestation} dissertation, where he also
related deforestation to the cut-elimination principle of
logic. \citet{Chin:1994:Safe} has also generalised Wadler's
deforestation to higher-order functional programs by using syntactic
properties to decide which terms that can be fused.

Both \citet{Hamilton:1995:Higher} and
\citet{Marlow:1995:Deforestation} have proven that their deforestation
algorithms terminate.  More recent work by
\citet{Hamilton:2006:Higher} extends deforestation with a treeless
form that is easy to recognise and handles a wide range of functions,
giving more transparency for the programmer.

\citet{Alimarine:2005:Improved} have improved the producer and
consumer analyses in Chin's \citeyearpar{Chin:1994:Safe} algorithm to
be based on semantics rather than syntax. They show that their algorithm
can remove much of the overhead introduced by generic programming
\citep{Hinze:2000:Generic}.

While these works are algorithmically rather close to ours due to the
close relationship between deforestation and positive supercompilation, it supposes either a call-by-name or call-by-need context, and is thus not applicable to the kind of languages we target.

\subsection{Supercompilation}

Closely related to deforestation is {\em supercompilation}
\citep{Turchin:1979:ASupercompiler, Turchin:1980:SemanticDefinitions,
  Turchin:1986:ProgramTransformation,
  Turchin:1986:TheConcept}. Supercompilation both removes
intermediate structures and achieves partial evaluation, as well as some
other optimisations. In partial evaluation terminology, the decision
of when to inline is taken online. The initial studies on
supercompilation were for the functional language Refal
\citep{Turchin:1989:Refal-5}. The supercompiler Scp4 \citep{Nemytykh:2003:Supercompiler} is implemented in Refal and is the most well-known implementation from this line of work. 

The {\em positive supercompiler} \citep{Soerensen:1996:Positive} is a
variant which only propagates positive information such as inferred equalities between terms. The propagation is done by unification and the work
highlights how similar deforestation and positive supercompilation
really are. Narrowing-driven partial evaluation
\citep{Alpuente:1998:Partial,Albert:2001:The} is the functional logic
programming equivalent of positive supercompilation but
formulated as a term rewriting system. Their approach also deals with
non-determinism from backtracking, which makes the corresponding algorithms more
complicated.

Strengthening the information propagation mechanism to propagate not
only positive, but also negative information, yields {\em perfect
supercompilation}
\citep{Secher:1999:Masters,Secher:1999:Perfect}. Negative information
is the opposite of positive information, namely inequalities. These
inequalities can be used to prune case-expression branches known not to be applicable, for example.

More recently, \citet{Mitchell:2008:ASupercompiler} have worked on
supercompiling Haskell. They report run-time reductions of up to 55\%
when their supercompiler is used in conjunction with GHC.

Supercompilation has seen applications beyond program optimization:
verification of cache coherence protocols
\citep{Lisitsa:2007:Verification} and proving term equivalence
\citep{Klyuchnikov:2009:Proving} are two examples. We do not believe
that our supercompiler is useful for these applications since it is inherently weaker
than the corresponding supercompiler with call-by-name semantics.

The positive supercompiler by \citet{Soerensen:1996:Positive} is the
immediate ancestor of our work, although we have extended it to a
higher-order language and converted it to work correctly for call-by-value
languages.

\subsection{Generalized Partial Computation}

GPC \citep{Futamura:1988:Generalized,Takano:1991:GeneralizedPartial} uses a theorem prover to extract
additional properties about the program being specialized. Among these
properties are the logical structure of a program, axioms for abstract
data types, and algebraic properties of primitive functions.

The theorem prover is applied whenever a
test is encountered, in order to determine which subset of the execution branches can actually be taken. Information about the predicate that was
tested is propagated along the branches that are left in the resulting
program. The reason GPC is such a powerful transformation is because
it assumes the unlimited power of a theorem prover. 

\citet{Futamura:2002:Program} have applied GPC in a call-by-value
setting in a system called WSDFU (Waseda
Simplify-Distribute-Fold-Unfold), and report many successful
experiments where optimal or near optimal residual programs are
produced. It is unclear whether WSDFU preserves termination behavior
or if it is a call-by-name transformation applied to a call-by-value
language.

We note that the rules for the first order language presented by
\citet{Takano:1991:GeneralizedPartial} are very similar to the
positive supercompiler, but the requirement for a theorem prover might
exclude the technique as a candidate for automatic compiler
optimisations. The lack of termination guarantees for the
transformation might be another obstacle. Considering the similarities
between GPC and positive supercompilation it should be straightforward to convert GPC to a call-by-value setting.

\subsection{Other Transformations}

Considering the vast amount of research conducted on program
transformations in general, we only briefly survey other related
transformations. 

\subsubsection{Partial Evaluation}

Partial evaluation \citep{Jones:1993:PartialEvaluation} is another
instance of Burstall and Darlington's
\citeyearpar{Burstall:77:ATransformation} informal class of
fold/unfold transformations. 

If partial evaluation is performed offline, the process is guided
by program annotations that tell when to fold, unfold, instantiate
and define functions. Binding-Time Analysis (BTA) is a program analysis that
annotates operations in the input program based on whether they are
statically known or not. 

Partial evaluation does not remove intermediate structures, something
we deem necessary to enable the programmer to write programs in the
clear and concise listful style.
Both deforestation and supercompilation simulate call-by-name
evaluation in the transformer, whereas partial evaluation simulates
call-by-value. It is suggested by
\citet{Soerensen:1994:TowardsUnifying} that this might affect the
strength of the transformation.

\subsubsection{Short Cut Fusion}

\label{sec:shortcut}

Short cut deforestation \citep{Gill:1993:AShortCut, Gill:1996:Cheap}
takes a different approach to deforestation, sacrificing some
generality by only working on lists. 

The idea is that the constructors {\em Nil} and {\em Cons} can be
replaced by a {\em foldr} consumer, and a special function {\em build}
is used to enable the transformation to recognize the producer and
enforce a type requirement. Lists using {\em build/foldr} can easily
be removed with the {\em foldr/build} rule:
\begin{haskell*}
foldr f c (build g) = g f c
\end{haskell*}

It is the responsibility of the programmer or compiler writer to make sure
list-traversing functions are written using {\em build} and {\em
  foldr}, thereby cluttering the code with information for the
optimiser and making it harder to read and understand for humans.

Gill implemented and measured short cut deforestation in GHC using the
nofib benchmark suite \citep{Partain:1992:nofib}. Around a dozen
benchmarks improved by more than 5\%, the average was 3\% and only one
example got noticeably worse, by 1\%. Heap allocations were reduced,
by half in one particular case. 

The main argument for short cut deforestation is its simplicity on the
compiler side compared to full-blown deforestation. GHC currently
contains a variant of the short cut deforestation implemented using
 rewrite rules \citep{Jones:2001:Playing}.

\citet{Takano:1995:Shortcut} generalized short cut
deforestation to work for any algebraic datatype through the acid rain
theorem. \citet{Ghani:2008:Short} have also generalized the {\em foldr/build} rule
to a {\em fold/superbuild} rule that can eliminate intermediate
structures of inductive types without disturbing the contexts in which
they are situated.

\citet{Launchbury:1995:Warm} worked on
automatically transforming programs into suitable form for shortcut
deforestation. \citet{Onoue:1997:Calculational} showed an
implementation of the acid rain theorem for Gofer where they could
automatically transform recursive functions into a form suitable for
shortcut fusion.

\citet{Chitil:2000:Type} used type-inference to transform the producer
of lists into the abstracted form required by short cut deforestation.
Given a type-inference algorithm which infers the most general type,
Chitil is able to determine the list constructors that need to be
replaced in one pass.

From the principal type property of the type inference algorithm
Chitil was also able to deduce completeness of the list abstraction
algorithm. This completeness guarantees that if a list can be
abstracted from a producer by abstracting its list constructors, then
the list abstraction algorithm will do so.

The implications of the completeness of the list abstraction algorithm
is that a {\it foldr} consumer can be fused with nearly any producer. One
reason list constructors might not be abstractable from a producer is
that they do not occur in the producer expression but in the
definition of a function which is called by the producer. A
worker/wrapper scheme proposed by Chitil ensures that these list constructors
are moved to the producer in order to make list abstraction possible.

The completeness property and the fact that the programmer does not
have to write any special code, in combination with the promising
results from measurements, suggest that short cut deforestation based on type-inference is a practical optimisation.

\citet{Takano:1995:Shortcut} noted that the {\it foldr}/{\it build} rule for short
cut deforestation has a dual. This is the {\em destroy/unfoldr} rule
used in Zip Fusion \citep{Svenningsson:2002:Shortcut}, which has some
interesting properties: it can remove all argument lists from a
function which consumes more than one list. The method described by
Svenningsson removes all intermediate lists in \<zip [1..n]
 [1..n]\>, addressing one of the main criticisms against the {\em
  foldr/build} rule.  The technique can also remove intermediate lists
from functions which consume their lists using accumulating parameters,
which is usually a problematic case. The {\em destroy/unfoldr} rule
is defined as:
\begin{haskell}
destroy g (unfoldr psi e) = g psi e
\end{haskell}

The Zip Fusion method is simple, and can be implemented in the same way as short cut
deforestation. It still suffers from the drawback that the programmer
or compiler writer has to make sure the list traversing functions are
written using {\em destroy} and {\em unfoldr}.

In more recent work \citet{Coutts:2007:Stream} have extended these
techniques to work on functions that handle nested lists, list
comprehensions and filter-like functions.

\section{Conclusions}
We have presented a positive supercompiler for a higher-order
call-by-value language and proven it correct with respect to
call-by-value semantics.
The adjustments required to preserve the termination properties of call-by-value evaluation are new and work well for many examples in the literature intended to show the usefulness of call-by-name transformations. 

\subsection{Future Work}
We believe that the linearity restriction of rule R14 is not necessary
for the soundness of our transformation, but have not yet found a way to prove
this. This is a natural topic for future work, as is an investigation of whether the concept of an {\it inlining budget} may be used to control the balance between supercompilation benefits and code size.

More work could be done on the strictness analysis component of our
supercompiler. We do not intend to focus on that subject, though;
instead we hope that the modular dependency on strictness analysis
will allow our supercompiler to readily take advantage of general
improvements in the area.

The supercompiler described in this article can be said to supersede
several of the standard transformations commonly implemented by
optimizing compilers, such as copy propagation, constant folding and
basic inlining. We conjecture that this range could be extended to
include transformations like common subexpression elimination as well,
by means of moderately small algorithm changes. An investigation of
the scope for such generalizations is an important area of future
research.

\section*{Acknowledgements}
The authors would like to thank Simon Marlow, Duncan Coutts and Neil
Mitchell for valuable discussions.
Thorsten Altenkirch contributed insights about non-termination. We
would also like to thank Viktor Leijon and the anonymous referees for
POPL'09 for providing useful comments that helped improve the
presentation and contents, and Germ{\'a}n Vidal for explaining
narrowing-driven partial evaluation to us.


\bibliographystyle{alpha}
\bibliography{pj,partial-eval}

\begin{thebibliography}{57}
\providecommand{\natexlab}[1]{#1}
\providecommand{\url}[1]{\texttt{#1}}
\expandafter\ifx\csname urlstyle\endcsname\relax
  \providecommand{\doi}[1]{doi: #1}\else
  \providecommand{\doi}{doi: \begingroup \urlstyle{rm}\Url}\fi

\bibitem[Albert and Vidal(2001)]{Albert:2001:The}
E.~Albert and G.~Vidal.
\newblock The narrowing-driven approach to functional logic program
  specialization.
\newblock \emph{New Generation Comput}, 20\penalty0 (1):\penalty0 3--26, 2001.

\bibitem[Alimarine and Smetsers(2005)]{Alimarine:2005:Improved}
A.~Alimarine and S.~Smetsers.
\newblock Improved fusion for optimizing generics.
\newblock In Manuel~V. Hermenegildo and Daniel Cabeza, editors, \emph{Practical
  Aspects of Declarative Languages, 7th International Symposium, {PADL} 2005,
  Long Beach, {CA}, {USA}, January 10-11, 2005, Proceedings}, volume 3350 of
  \emph{Lecture Notes in Computer Science}, pages 203--218. Springer, 2005.
\newblock ISBN 3-540-24362-3.

\bibitem[Alpuente et~al.(1998)Alpuente, Falaschi, and
  Vidal]{Alpuente:1998:Partial}
M.~Alpuente, M.~Falaschi, and G.~Vidal.
\newblock {P}artial {E}valuation of {F}unctional {L}ogic {P}rograms.
\newblock \emph{ACM Transactions on Programming Languages and Systems},
  20\penalty0 (4):\penalty0 768--844, 1998.

\bibitem[Burstall and Darlington(1977)]{Burstall:77:ATransformation}
R.M. Burstall and J.~Darlington.
\newblock A transformation system for developing recursive programs.
\newblock \emph{Journal of the {ACM}}, 24\penalty0 (1):\penalty0 44--67,
  January 1977.

\bibitem[Chin(1994)]{Chin:1994:Safe}
W-N. Chin.
\newblock Safe fusion of functional expressions {II}: Further improvements.
\newblock \emph{J. Funct. Program}, 4\penalty0 (4):\penalty0 515--555, 1994.

\bibitem[Chitil(2000)]{Chitil:2000:Type}
O.~Chitil.
\newblock \emph{Type-Inference Based Deforestation of Functional Programs}.
\newblock PhD thesis, RWTH Aachen, October 2000.

\bibitem[Coutts et~al.(2007)Coutts, Leshchinskiy, and
  Stewart]{Coutts:2007:Stream}
D.~Coutts, R.~Leshchinskiy, and D.~Stewart.
\newblock Stream fusion: from lists to streams to nothing at all.
\newblock In \emph{ICFP '07: Proceedings of the 12th ACM SIGPLAN international
  conference on Functional programming}, pages 315--326, New York, NY, USA,
  2007. ACM.
\newblock ISBN 978-1-59593-815-2.

\bibitem[Dershowitz(1987)]{Dershowitz:1987:Termination}
N.~Dershowitz.
\newblock Termination of rewriting.
\newblock \emph{Journal of Symbolic Computation}, 3\penalty0 (1):\penalty0
  69--115, 1987.

\bibitem[Futamura and Nogi(1988)]{Futamura:1988:Generalized}
Y.~Futamura and K.~Nogi.
\newblock Generalized partial computation.
\newblock In D.~Bj{\o}rner, A.P. Ershov, and N.D. Jones, editors, \emph{Partial
  Evaluation and Mixed Computation}, pages 133--151. Amsterdam:\ North-Holland,
  1988.

\bibitem[Futamura et~al.(2002)Futamura, Konishi, and
  Gl{\"u}ck]{Futamura:2002:Program}
Y.~Futamura, Z.~Konishi, and R.~Gl{\"u}ck.
\newblock Program transformation system based on generalized partial
  computation.
\newblock \emph{New Gen. Comput.}, 20\penalty0 (1):\penalty0 75--99, 2002.
\newblock ISSN 0288-3635.

\bibitem[Ghani and Johann(2008)]{Ghani:2008:Short}
N.~Ghani and P.~Johann.
\newblock Short cut fusion of recursive programs with computational effects.
\newblock In P.~{Achten}, P.~{Koopman}, and M.~T. {Moraz{\'a}n}, editors,
  \emph{{Draft Proceedings of The Ninth Symposium on Trends in Functional
  Programming (TFP)}}, number ICIS--R08007, 2008.

\bibitem[Gill et~al.(1993)Gill, Launchbury, and
  Peyton~Jones]{Gill:1993:AShortCut}
A.~Gill, J.~Launchbury, and S.L. Peyton~Jones.
\newblock A short cut to deforestation.
\newblock In \emph{Functional Programming Languages and Computer Architecture,
  Copenhagen, Denmark, 1993}, 1993.

\bibitem[Gill(1996)]{Gill:1996:Cheap}
A.~J. Gill.
\newblock \emph{Cheap Deforestation for Non-strict Functional Languages}.
\newblock PhD thesis, Univ.\ of Glasgow, January 1996.

\bibitem[Hamilton(1996)]{Hamilton:1995:Higher}
G.~W. Hamilton.
\newblock Higher order deforestation.
\newblock In \emph{PLILP '96: Proceedings of the 8th International Symposium on
  Programming Languages: Implementations, Logics, and Programs}, pages
  213--227, London, UK, 1996. Springer-Verlag.
\newblock ISBN 3-540-61756-6.

\bibitem[Hamilton(2006)]{Hamilton:2006:Higher}
G.~W. Hamilton.
\newblock Higher order deforestation.
\newblock \emph{Fundam. Informaticae}, 69\penalty0 (1-2):\penalty0 39--61,
  2006.

\bibitem[Hinze(2000)]{Hinze:2000:Generic}
R.~Hinze.
\newblock \emph{{Generic Programs and Proofs}}.
\newblock Habilitationsschrift, Bonn University, 2000.

\bibitem[Johnsson(1985)]{Johnsson:1985:Lambda}
T.~Johnsson.
\newblock Lambda lifting: Transforming programs to recursive equations.
\newblock In \emph{FPCA}, pages 190--203, 1985.

\bibitem[Jones et~al.(1993)Jones, Gomard, and
  Sestoft]{Jones:1993:PartialEvaluation}
N.D. Jones, C.K. Gomard, and P.~Sestoft.
\newblock \emph{Partial Evaluation and Automatic Program Generation}.
\newblock Englewood Cliffs, NJ:\ Prentice Hall, 1993.
\newblock ISBN 0-13-020249-5.

\bibitem[Jonsson(2008)]{Jonsson:2008:Positive}
P.~A. Jonsson.
\newblock Positive supercompilation for a higher-order call-by-value language.
\newblock Licentiate thesis, Lule{\aa} University of Technology, Sweden, Jun
  2008.

\bibitem[Jonsson and Nordlander(2009)]{Jonsson:2009:Positive}
P.~A. Jonsson and J.~Nordlander.
\newblock Positive supercompilation for a higher-order call-by-value language.
\newblock In \emph{POPL '09: Proceedings of the 36th annual ACM SIGPLAN-SIGACT
  symposium on Principles of programming languages}, 2009.

\bibitem[Klyuchnikov and Romanenko(2009)]{Klyuchnikov:2009:Proving}
I.~Klyuchnikov and S.~Romanenko.
\newblock Proving the equivalence of higher-order terms by means of
  supercompilation.
\newblock In \emph{PSI '09: Proceedings of the Seventh International Andrei
  Ershov Memorial Conference}, 2009.

\bibitem[Kort(1996)]{Kort:1996:Deforestation}
J.~Kort.
\newblock {Deforestation of a raytracer}.
\newblock Master's thesis, University of Amsterdam, 1996.

\bibitem[Lassez et~al.(1988)Lassez, Maher, and
  Marriott]{Lassez:1988:Unification}
J-L. Lassez, M.~Maher, and K.~Marriott.
\newblock Unification revisited.
\newblock In Jack Minker, editor, \emph{Foundations of Deductive Databases and
  Logic Programming}, pages 587--625. Morgan Kaufmann, 1988.

\bibitem[Launchbury and Sheard(1995)]{Launchbury:1995:Warm}
J.~Launchbury and T.~Sheard.
\newblock Warm fusion: Deriving build-cata's from recursive definitions.
\newblock In \emph{FPCA}, pages 314--323, 1995.

\bibitem[Leroy(2008)]{Ocaml}
X.~Leroy.
\newblock The {O}bjective {C}aml system: Documentation and user's manual, 2008.
\newblock With D. Doligez, J. Garrigue, D. R\'emy, and J. Vouillon. Available
  from \url{http://caml.inria.fr} (1996--2008).

\bibitem[Lisitsa and Nemytykh(2007)]{Lisitsa:2007:Verification}
A.~Lisitsa and A.~P. Nemytykh.
\newblock Verification as a parameterized testing (experiments with the {SCP4}
  supercompiler).
\newblock \emph{Programming and Computer Software}, 33\penalty0 (1):\penalty0
  14--23, 2007.

\bibitem[Marlow and Wadler(1992)]{Marlow:1992:Deforestation}
S.~Marlow and P.~Wadler.
\newblock Deforestation for higher-order functions.
\newblock In John Launchbury and Patrick~M. Sansom, editors, \emph{Functional
  Programming}, Workshops in Computing, pages 154--165. Springer, 1992.
\newblock ISBN 3-540-19820-2.

\bibitem[Marlow(1995)]{Marlow:1995:Deforestation}
S.~D. Marlow.
\newblock \emph{Deforestation for Higher-Order Functional Programs}.
\newblock PhD thesis, Department of Computing Science, University of Glasgow,
  April~27 1995.

\bibitem[Milner et~al.(1997)Milner, Tofte, Harper, and MacQueen]{SML97}
R.~Milner, M.~Tofte, R.~Harper, and D.~MacQueen.
\newblock \emph{The Definition of {S}tandard {ML}, {\rm Revised edition}}.
\newblock MIT Press, 1997.

\bibitem[Mitchell and Runciman(2008)]{Mitchell:2008:ASupercompiler}
N.~Mitchell and C.~Runciman.
\newblock A supercompiler for core {H}askell.
\newblock In O.~Chitil et~al., editor, \emph{Selected Papers from the
  Proceedings of IFL 2007}, volume 5083 of \emph{Lecture Notes in Computer
  Science}, pages 147--164. Springer-Verlag, 2008.

\bibitem[Narendran and Stillman(1987)]{Narendran:1987:On}
P.~Narendran and J.~Stillman.
\newblock On the {C}omplexity of {H}omeomorphic {E}mbeddings.
\newblock Technical Report 87-8, Computer Science Department, State Univeristy
  of New York at Albany, March 1987.

\bibitem[Nash-Williams(1963)]{Nash-Williams:1963:On}
C.~St. J.~A. Nash-Williams.
\newblock On well-quasi-ordering finite trees.
\newblock \emph{Proceedings of the Cambridge Philosophical Society},
  59\penalty0 (4):\penalty0 833--835, October 1963.

\bibitem[Nemytykh(2003)]{Nemytykh:2003:Supercompiler}
A.~P. Nemytykh.
\newblock The supercompiler {SCP4}: General structure.
\newblock In Manfred Broy and Alexandre~V. Zamulin, editors, \emph{Perspectives
  of Systems Informatics, 5th International Andrei Ershov Memorial Conference,
  {PSI} 2003, Akademgorodok, Novosibirsk, Russia, July 9-12, 2003, Revised
  Papers}, volume 2890 of \emph{LNCS}, pages 162--170. Springer, 2003.
\newblock ISBN 3-540-20813-5.

\bibitem[Nordlander et~al.(2008)Nordlander, Carlsson, Gill, Lindgren, and {von
  Sydow}]{Timber}
J.~Nordlander, M.~Carlsson, A.~Gill, P.~Lindgren, and B.~{von Sydow}.
\newblock The {T}imber home page, 2008.
\newblock URL \url{http://www.timber-lang.org}.

\bibitem[Ohori and Sasano(2007)]{Ohori:2007:Lightweight}
A.~Ohori and I.~Sasano.
\newblock Lightweight fusion by fixed point promotion.
\newblock In \emph{POPL '07: Proceedings of the 34th annual ACM SIGPLAN-SIGACT
  symposium on Principles of programming languages}, pages 143--154, New York,
  NY, USA, 2007. ACM.
\newblock ISBN 1-59593-575-4.

\bibitem[Onoue et~al.(1997)Onoue, Hu, Iwasaki, and
  Takeichi]{Onoue:1997:Calculational}
Y.~Onoue, Z.~Hu, H.~Iwasaki, and M.~Takeichi.
\newblock A calculational fusion system {HYLO}.
\newblock In R.~S. Bird and L.~G. L.~T. Meertens, editors, \emph{Algorithmic
  Languages and Calculi, {IFIP} {TC2} {WG2}.1 International Workshop on
  Algorithmic Languages and Calculi, 17-22 February 1997, Alsace, France},
  volume~95 of \emph{IFIP Conference Proceedings}, pages 76--106. Chapman \&
  Hall, 1997.
\newblock ISBN 0-412-82050-1.

\bibitem[Partain(1992)]{Partain:1992:nofib}
W.~Partain.
\newblock The nofib benchmark suite of {H}askell programs.
\newblock In John Launchbury and Patrick~M. Sansom, editors, \emph{Functional
  Programming}, Workshops in Computing, pages 195--202. Springer, 1992.
\newblock ISBN 3-540-19820-2.

\bibitem[{Peyton Jones} et~al.(2001){Peyton Jones}, Tolmach, and
  Hoare]{Jones:2001:Playing}
S.~L. {Peyton Jones}, A.~Tolmach, and T.~Hoare.
\newblock Playing by the rules: Rewriting as a practical optimisation technique
  in {GHC}.
\newblock In Ralf Hinze, editor, \emph{Proceedings of the 2001 ACM SIGPLAN
  Haskell Workshop (HW'2001), 2nd September 2001, Firenze, Italy.}, Electronic
  Notes in Theoretical Computer Science, Vol 59. Utrecht University,
  September~28 2001.
\newblock UU-CS-2001-23.

\bibitem[Sands(1996{\natexlab{a}})]{Sands:1996:Proving}
D.~Sands.
\newblock Proving the correctness of recursion-based automatic program
  transformations.
\newblock \emph{Theoretical Computer Science}, 167\penalty0 (1--2):\penalty0
  193--233, 30~October 1996{\natexlab{a}}.

\bibitem[Sands(1996{\natexlab{b}})]{Sands:1996:Total}
D.~Sands.
\newblock Total correctness by local improvement in the transformation of
  functional programs.
\newblock \emph{ACM Transactions on Programming Languages and Systems},
  18\penalty0 (2):\penalty0 175--234, March 1996{\natexlab{b}}.

\bibitem[Sands(1997)]{Sands:1997:From}
D.~Sands.
\newblock From {SOS} rules to proof principles: An operational metatheory for
  functional languages.
\newblock In \emph{Proceedings of the 24th Annual ACM SIGPLAN-SIGACT Symposium
  on Principles of Programming Languages (POPL)}. ACM Press, January 1997.

\bibitem[Secher(1999)]{Secher:1999:Masters}
J.~P. Secher.
\newblock Perfect supercompilation.
\newblock Technical Report DIKU-TR-99/1, Department of Computer Science (DIKU),
  University of Copenhagen, February 1999.

\bibitem[Secher and S{\o}rensen(2000)]{Secher:1999:Perfect}
J.P. Secher and M.H. S{\o}rensen.
\newblock On perfect supercompilation.
\newblock In D.~Bj{\o}rner, M.~Broy, and A.~Zamulin, editors, \emph{Proceedings
  of Perspectives of System Informatics}, volume 1755 of \emph{Lecture Notes in
  Computer Science}, pages 113--127. Springer-Verlag, 2000.

\bibitem[S{\o}rensen(2000)]{Soerensen:2000:Convergence}
M.H. S{\o}rensen.
\newblock Convergence of program transformers in the metric space of trees.
\newblock \emph{Sci. Comput. Program}, 37\penalty0 (1-3):\penalty0 163--205,
  2000.

\bibitem[S{\o}rensen and Gl{\"u}ck(1995)]{Soerensen:1995:AnAlgorithm}
M.H. S{\o}rensen and R.~Gl{\"u}ck.
\newblock An algorithm of generalization in positive supercompilation.
\newblock In J.W. Lloyd, editor, \emph{International Logic Programming
  Symposium}, pages 465--479. Cambridge, MA:\ MIT Press, 1995.

\bibitem[S{\o}rensen et~al.(1994)S{\o}rensen, Gl{\"u}ck, and
  Jones]{Soerensen:1994:TowardsUnifying}
M.H. S{\o}rensen, R.~Gl{\"u}ck, and N.D. Jones.
\newblock Towards unifying partial evaluation, deforestation, supercompilation,
  and {GPC}.
\newblock In D.~Sannella, editor, \emph{Programming Languages and Systems ---
  ESOP'94. 5th European Symposium on Programming, Edinburgh, U.K., April 1994
  (Lecture Notes in Computer Science, vol. 788)}, pages 485--500. Berlin:\
  Springer-Verlag, 1994.

\bibitem[S{\o}rensen et~al.(1996)S{\o}rensen, Gl{\"u}ck, and
  Jones]{Soerensen:1996:Positive}
M.H. S{\o}rensen, R.~Gl{\"u}ck, and N.D. Jones.
\newblock A positive supercompiler.
\newblock \emph{Journal of Functional Programming}, 6\penalty0 (6):\penalty0
  811--838, 1996.

\bibitem[Svenningsson(2002)]{Svenningsson:2002:Shortcut}
J.~Svenningsson.
\newblock Shortcut fusion for accumulating parameters \& zip-like functions.
\newblock In \emph{ICFP}, pages 124--132, 2002.

\bibitem[Syme(2008)]{Syme:2008:The}
D.~Syme.
\newblock The {F}\# programming language, June 2008.
\newblock URL \url{http://research.microsoft.com/fsharp}.

\bibitem[Takano(1991)]{Takano:1991:GeneralizedPartial}
A.~Takano.
\newblock Generalized partial computation for a lazy functional language.
\newblock In \emph{Partial Evaluation and Semantics-Based Program Manipulation,
  New Haven, Connecticut (Sigplan Notices, vol. 26, no. 9, September 1991)},
  pages 1--11. New York:\ ACM, 1991.

\bibitem[Takano and Meijer(1995)]{Takano:1995:Shortcut}
A.~Takano and E.~Meijer.
\newblock Shortcut deforestation in calculational form.
\newblock In \emph{FPCA}, pages 306--313, 1995.

\bibitem[Turchin(1979)]{Turchin:1979:ASupercompiler}
V.F. Turchin.
\newblock A supercompiler system based on the language {R}efal.
\newblock \emph{SIGPLAN Notices}, 14\penalty0 (2):\penalty0 46--54, February
  1979.

\bibitem[Turchin(1980)]{Turchin:1980:SemanticDefinitions}
V.F. Turchin.
\newblock Semantic definitions in {R}efal and automatic production of
  compilers.
\newblock In N.D. Jones, editor, \emph{Semantics-Directed Compiler Generation,
  Aarhus, Denmark (Lecture Notes in Computer Science, vol. 94)}, pages
  441--474. Berlin:\ Springer-Verlag, 1980.

\bibitem[Turchin(1986{\natexlab{a}})]{Turchin:1986:ProgramTransformation}
V.F. Turchin.
\newblock Program transformation by supercompilation.
\newblock In H.~Ganzinger and N.D. Jones, editors, \emph{Programs as Data
  Objects, Copenhagen, Denmark, 1985 (Lecture Notes in Computer Science, vol.
  217)}, pages 257--281. Berlin:\ Springer-Verlag, 1986{\natexlab{a}}.

\bibitem[Turchin(1986{\natexlab{b}})]{Turchin:1986:TheConcept}
V.F. Turchin.
\newblock The concept of a supercompiler.
\newblock \emph{ACM Transactions on Programming Languages and Systems},
  8\penalty0 (3):\penalty0 292--325, July 1986{\natexlab{b}}.

\bibitem[Turchin(1989)]{Turchin:1989:Refal-5}
V.F. Turchin.
\newblock \emph{Refal-5, Programming Guide {\&} Reference Manual}.
\newblock Holyoke, MA:\ New England Publishing Co., 1989.

\bibitem[Wadler(1990)]{Wadler:1990:Deforestation}
P.~Wadler.
\newblock Deforestation: transforming programs to eliminate trees.
\newblock \emph{Theoretical Computer Science}, 73\penalty0 (2):\penalty0
  231--248, June 1990.
\newblock ISSN 0304-3975.

\end{thebibliography}

\appendix\overfullrule=2 pt

\section{Proofs}
\label{sec:Voj}

We borrow a couple of technical lemmas from
\citet{Sands:1996:Proving}, and adapt the proofs to be valid under call-by-value:

\begin{lem}[Sands, p. 24]
\label{lem:letrec}
For all expressions e and value substitutions $\theta$ such that $h
\notin dom(\theta)$, if $e_0 \eval^{1} e_1$ then
\begin{equation*}
  \LETREC h = \lambda \overline{x}.e_1 \IN [\theta (e_0)/z]e \costeq
  \LETREC h = \lambda \overline{x}.e_1 \IN [h\, \theta (\overline{x})/z]e
\end{equation*}
\end{lem}

\begin{proof}
Expanding both sides according to the definition of letrec yields: 
\begin{equation*}
(\lambda h.[\theta (e_0)/z]e)\,(\lambda n.fix\, (\lambda h.\lambda
\overline{x}.e_1)\, n) \costeq (\lambda h.[h\, \theta
(\overline{x})/z]e)\,(\lambda n.fix\, (\lambda h.\lambda
\overline{x}.e_1)\, n) 
\end{equation*}
and evaluating both sides one step $\eval$ gives:
\begin{equation*}
[\lambda n.fix\, (\lambda h.\lambda \overline{x}.e_1)\, n/h][\theta
(e_0)/z]e \costeq [\lambda n.fix\, (\lambda h.\lambda
\overline{x}.e_1)\, n/h][h\, \theta (\overline{x})/z]e
\end{equation*}
From this we can see that it is sufficient to prove:
\begin{equation*}
[\lambda n.fix\, (\lambda h.\lambda \overline{x}.e_1)\, n/h]\theta
e_0 \costeq [\lambda n.fix\, (\lambda h.\lambda \overline{x}.e_1)\,
n/h]h\, \theta (\overline{x})
\end{equation*}
The substitution $\theta$ can safely be moved out since $h \notin dom(\theta)$:
\begin{equation*}
[\lambda n.fix\, (\lambda h.\lambda \overline{x}.e_1)\, n/h]\theta
e_0 \costeq [\lambda n.fix\, (\lambda h.\lambda \overline{x}.e_1)\,
n/h] \theta(h\, \overline{x})
\end{equation*}
Performing evaluation steps on both sides  yield:
\begin{align*}
[\lambda n.fix\, (\lambda h.\lambda \overline{x}.e_1)\, n/h]\theta e_0  \costeq &[\lambda n.fix\, (\lambda h.\lambda \overline{x}.e_1)\,n/h]\theta (h\,  \overline{x}) \\
[\lambda n.fix\, (\lambda h.\lambda \overline{x}.e_1)\, n/h]\theta e_0  \costeq & [\lambda n.fix\, (\lambda h.\lambda \overline{x}.e_1)\,n/h]\theta ((\lambda n.fix\, (\lambda h.\lambda \overline{x}.e_1)\,n)\, \overline{x}) \\
[\lambda n.fix\, (\lambda h.\lambda \overline{x}.e_1)\, n/h]\theta e_0  \costeq & [\lambda n.fix\, (\lambda h.\lambda \overline{x}.e_1)\,n/h]\theta (fix\, (\lambda h.\lambda \overline{x}.e_1)\,\overline{x})  \\
[\lambda n.fix\, (\lambda h.\lambda \overline{x}.e_1)\, n/h]\theta e_1  \costeq & [\lambda n.fix\, (\lambda h.\lambda \overline{x}.e_1)\,n/h]\theta ((\lambda f.f\,(\lambda n.fix\,f\,n))\, (\lambda h.\lambda \overline{x}.e_1)\,\overline{x})  \\
[\lambda n.fix\, (\lambda h.\lambda \overline{x}.e_1)\, n/h]\theta e_1  \costeq & [\lambda n.fix\, (\lambda h.\lambda \overline{x}.e_1)\,n/h]\theta ((\lambda h.\lambda \overline{x}.e_1)\,(\lambda n.fix\,(\lambda h.\lambda \overline{x}.e_1)\,n)\,\overline{x})  \\
[\lambda n.fix\, (\lambda h.\lambda \overline{x}.e_1)\, n/h]\theta e_1  \costeq & [\lambda n.fix\, (\lambda h.\lambda \overline{x}.e_1)\,n/h]\theta ((\lambda \overline{x}.e_1)\,\overline{x})  \\
[\lambda n.fix\, (\lambda h.\lambda \overline{x}.e_1)\, n/h]\theta e_1  \costeq &[\overline{x}/\overline{x}][\lambda n.fix\, (\lambda h.\lambda \overline{x}.e_1)\,n/h]\theta e_1  \\
[\lambda n.fix\, (\lambda h.\lambda \overline{x}.e_1)\, n/h]\theta e_1  \costeq& [\lambda n.fix\, (\lambda h.\lambda \overline{x}.e_1)\,n/h]\theta e_1  \\
\end{align*}
The LHS and the RHS are cost equivalent, so by Lemma \ref{lem:steps}
the initial expressions are cost equivalent.
\end{proof}

\begin{lem}[Sands, p. 25]
\label{lem:drp}
$\rho'(\drp{\con{v}}{\boxempty}) \costeq \LETREC h = \lambda  \overline{x} .\con{v} \IN
\rho(\drp{\con{v}}{\boxempty})$
\end{lem}

\proof [Proof (Similar to \citet{Sands:1996:Proving})]
  By inspection of the rules for \<\sdr \edr\>, all free occurrences of
  {\em h} in $\drp{\con{v}}{\boxempty}$ must occur in sub-expressions of the form
  $h\, \overline{x}$. Suppose there are {\em k} such occurrences, which
  we can write as $\theta_1 h\,\overline{x}\ldots \theta_k
  h\,\overline{x}$, where the $\theta_i$ are just renamings of the
  variables $\overline{x}$. So $\drp{\con{v}}{\boxempty}$ can be written as
  $[\theta_1 h\,\overline{x}\ldots \theta_k h\,\overline{x}/z_1\ldots
  z_k]e'$, where $e'$ contains no free occurrences of {\em h}. Then
  (substitution associates to the right):
\begin{eqnarray*}
  \rho'(\drp{\con{v}}{\boxempty}) & \equiv & [\lambda \overline{x}.\con{g}/h]\rho(\drp{\con{v}}{\boxempty})\\
  & \costeq & [\lambda \overline{x}.\con{g}/h]\rho([\theta_1
  h\,\overline{x}\ldots \theta_k h\,\overline{x} /z_1\ldots z_k]e') \\
  & \costeq & \rho([\theta_1 \con{g}\ldots \theta_k \con{g}/z_1\ldots z_k]e')  \\
  & \costeq & \text{ (by Lemma \ref{lem:letrec}) } \\
 & & \LETREC h = \lambda \overline{x} . \con{v\,  \overline{e}} \IN \rho([\theta_1
 h\,\overline{x}\ldots \theta_k h\,\overline{x} /z_1\ldots z_k]e') \\
 & \equiv & \LETREC h = \lambda \overline{x}.\con{v} \IN
  \rho(\drp{\con{v}}{\boxempty})\rlap{\hbox to95 pt{\hfil\qEd}} \\
\end{eqnarray*}

\begin{lem}
\label{lem:eqlet2}
$\con{\LET x = e \IN f} \costrighteq_s \LET x = e \IN \con{f}$
\end{lem}

\begin{proof}
  Notice that $\con{\LET x = \boxempty \IN f}$ is a redex, and assume $e
  \eval^k v$. The LHS evaluates in k steps $\con{\LET x = e \IN f}
  \eval^k \con{\LET x = v \IN f}$ $ \eval \con{[v/x]f}$, and the RHS
  evaluates in k steps $\LET x = e \IN \con{f} \eval^k \LET x = v \IN
  \con{f} \eval [v/x]\con{f}$. Since contexts do not bind variables
  these two terms are equivalent and by Lemma \ref{lem:steps} the
  initial terms are cost equivalent.
\end{proof}

\begin{lem}
\label{lem:eqletrec}
$\con{\LETREC g = v \IN e} \costrighteq_S \LETREC g = v \IN \con{e}$
\end{lem}

\begin{proof}
  Translate both sides by the definition of letrec into
  $\con{(\lambda g.e)\, (\lambda n.fix\, (\lambda g.v)\,n)}$ $
  \costrighteq_S (\lambda g.\con{e})\, (\lambda n.fix\, (\lambda
  g.v)\,n)$. Notice that $\con{\boxempty}$ is a redex. The LHS evaluates in 0
  steps to $\con{(\lambda g.e)\, (\lambda n.fix\, (\lambda g.v)\,n)}
  \eval \con{[\lambda n.fix\, (\lambda g.v)\,n/g]e}$ and the RHS evaluates in 0 steps to 
  $(\lambda g.\con{e})\, (\lambda n.fix\, (\lambda g.v)\,n)$ $ \eval
  [\lambda n.fix\, (\lambda g.v)\,n/g]\con{e}$. Since our contexts do
  not bind variables these two terms are equivalent and by Lemma
  \ref{lem:steps} the initial terms are cost equivalent.
\end{proof}

\begin{lem}
\label{lem:eqcase}
$\con{\CASE e \OF \{ p_i \to e_i \}} \costrighteq_s \CASE e \OF \{ p_i
\to \con{e_i} \}$
\end{lem}

\begin{proof}
  Notice that $\con{\CASE \boxempty \OF \{ p_i \to e_i \}}$ is a redex, and
  assume $e \eval^k n_j$. The LHS evaluates in k steps $\con{\CASE e
    \OF \{ p_i \to e_i \}} \eval^k \con{\CASE n_j \OF \{ p_i \to e_i
    \}} \eval \con{e_j}$, and the RHS evaluates in k steps $\CASE e
  \OF \{ p_i \to \con{e_i} \} \eval^k \CASE n_j \OF \{ p_i \to
  \con{e_i} \} \con{f} \eval \con{e_j}$. Since these two terms are
  equivalent the initial terms are cost equivalent by Lemma
  \ref{lem:steps}.
\end{proof}

We set out to prove the main theorem about total correctness:
\begin{thm}[Total Correctness]
Let $\con{e}$ be an expression, and $\rho$ an environment such that
\begin{enumerate}[$\bullet$]
\item the range of $\rho$ contains only closed expressions, and
\item $\fv(\con{e}) \cap dom(\rho) = \emptyset$, and
\end{enumerate}
then $\con{e} \costrighteq_s \rho(\dr{e}{\R})$.
\end{thm}
\noindent We reason by induction on
the structure of expressions, and since the algorithm is total (Lemma
\ref{lem:totality}) this coincides with inspection of each rule.

\subsection{R1}
\label{sec:startbev}

We have that $\rho(\dr{n}{\R}) = \rho(\con{n})$, and the conditions of
the proposition ensure that $\fv(\con{n}) \cap dom(\rho) = \emptyset$,
so $\rho(\con{n}) = \con{n}$ This is syntactically equivalent to the
input, and we conclude $\con{n} \costrighteq_s \rho(\dr{n}{\R})$.

\subsection{R2} 

We have that $\rho(\dr{x}{\R}) = \rho(\con{x})$, and the conditions of
the proposition ensure that $\fv(\con{x}) \cap dom(\rho) = \emptyset$,
so $\rho(\con{x}) = \con{x}$ This is syntactically equivalent to the
input, and we conclude $\con{x} \costrighteq_s \rho(\dr{x}{\R})$.

\subsection{R3}

\subsubsection{Case: (1)}

\subparagraph{Suppose $\exists h.\rho(h) \equiv \lambda
  \overline{x}.\con{g}$ and hence that $\dr{\con{g}}{\boxempty} = h\, \overline{x}$} ~\\

The conditions of the proposition ensure that $\overline{x} \cap
dom(\rho) = \emptyset$, so $\rho(\dr{\con{g}}{\boxempty}) = \rho(h\, \overline{x}) =
(\lambda \overline{x}.\con{g})\, \overline{x}$. However, $\con{g}$ and
$(\lambda \overline{x}.\con{g})\, \overline{x}$ are cost equivalent, which
implies strong improvement, and we conclude $\con{g} \costrighteq_s
\rho(\dr{\con{g}}{\boxempty})$ 

\subsubsection{Case: (2)}

\subparagraph{Suppose $\exists(h, t) \in \rho. t \costlefteq \con{g}$  and that $\con{g} \costlefteq t$,  hence $\dr{\con{g}}{\boxempty} = \con{g}$
} ~ \\

The term on the RHS is discarded and replaced with a new term higher
up in the tree, so it does not matter what the term is. 

\subsubsection{Case: (3)}

\subparagraph{Suppose $\exists(h, t) \in \rho. t \costlefteq \con{g}$ and hence that $\dr{\con{g}}{\boxempty} =
  [\dr{\overline{f}}{\boxempty}/\overline{y}] \dr{f_g}{\boxempty}$
} ~ \\

We have $\rho(\dr{g}{\R}) =
  \rho([\dr{\overline{f}}{\boxempty}/\overline{y}] \dr{f_g}{\boxempty})$ $ =
  [\rho(\dr{\overline{f}}{\boxempty})/\overline{x}]\rho(\dr{f_g}{\boxempty})$. By the induction
  hypothesis, $\overline{f} \costrighteq_s \rho(\dr{\overline{f}}{\boxempty}) $ and $f_g
  \costrighteq_s  \rho(\dr{f_g}{\boxempty})$
  and by congruence properties of strong improvement (Lemma
  \ref{lem:simplaws}:1) $\con{g} \costrighteq_s
  \rho(\dr{g}{\R})$.

\subsubsection{Case: (4a)}

Analogous to the previous case.

\subsubsection{Case: (4b)}

\subparagraph{If $\dr{\con{g}}{\boxempty} = \rho(\LETREC h = \lambda \overline{x} .  \drp{\con{v}}{\boxempty}
  \IN h\, \overline{x})$}~\\   
where $\rho' = \rho \cup  (h, \lambda \overline{x}.\con{g})$ and $h
\notin (\overline{x} \cup dom(\rho))$. We need to show that:

\begin{equation*}
\con{g} \costrighteq_s \rho(\LETREC h = \lambda \overline{x}. \drp{\con{v}}{\boxempty} \IN h\, \overline{x})
\end{equation*}

Since $h, \overline{x} \notin dom(\rho)$ we have that  $\rho(\LETREC h
= \lambda \overline{x} . \drp{\con{v}}{\boxempty} \IN h\, \overline{x}) \equiv 
\LETREC h = $ $\lambda \overline{x} .\rho(\drp{\con{v}}{\boxempty}) $ $\IN h\,
\overline{x}$. 

$\R$ is a reduction context, hence $\con{g} \eval^{1} \con{v}$. By Lemma
\ref{lem:letrec} we have that $\LETREC h = \lambda \overline{x} . \con{v} \IN
\con{g\, \overline{e}} \costeq \LETREC h = \lambda \overline{x} . \con{v} \IN h \,
\overline{x}$. Since $h \notin \fv(\con{g})$ this simplifies to $\con{g}
\costeq \LETREC h = \lambda \overline{x}. \con{v} \IN h \, \overline{x}$. It is
necessary and sufficient to prove that

\begin{equation*}
\LETREC h = \lambda  \overline{x}. \con{v} \IN h \, \overline{x} \costrighteq_s
\LETREC h = \lambda  \overline{x} . \rho(\drp{\con{v}}{\boxempty}) \IN h \, \overline{x}
\end{equation*}

\noindent By Theorem \ref{thm:letrec-simp} it is sufficient to show: 

\begin{equation*}
\LETREC h = \lambda   \overline{x} . \con{v} \IN \con{v} \costrighteq_s
\end{equation*}
\begin{equation*}
\LETREC h = \lambda  \overline{x} . \con{v} \IN \rho(\drp{\con{v}}{\boxempty}) 
\end{equation*}

\noindent By Lemma \ref{lem:drp} and $\LETREC h = \lambda
\overline{x}. \con{v} \IN \con{v} \costeq \con{v}$, this is equivalent to showing that

\begin{equation*}
\con{v} \costrighteq_s \rho'(\drp{\con{v}}{\boxempty})
\end{equation*}

\noindent Which follows from the induction hypothesis, since it is a shorter
transformation. 

\subsubsection{Case: (4c)}

We have that $\rho(\dr{g}{\R}) =
\rho(\dr{\con{v}}{\boxempty})$. By the induction hypothesis
$\con{v} \costrighteq_s
\rho(\dr{\con{v}}{\boxempty})$, and since
$\con{g} \eval^1 \con{v}$ it follows from
Lemma \ref{lem:simplaws}:4 that $\con{g} \costrighteq_s
\rho(\dr{g}{\R})$.

\subsection{R4}

We have that $\rho(\dr{k\,\overline{e}}{\boxempty}) =
\rho(k\,\dr{\overline{e}}{\boxempty})$, and the conditions of the proposition
ensure that $\fv(k\,\overline{e}) \cap dom(\rho) = \emptyset$, so
$\rho(k\,\dr{\overline{e}}{\boxempty}) = k\, \rho(\dr{\overline{e}}{\boxempty})$. By
the induction hypothesis, $\overline{e} \costrighteq_s
\rho(\dr{\overline{e}}{\boxempty})$, and by congruence properties of strong
improvement (Lemma \ref{lem:simplaws}:1) $k\,\overline{e}
\costrighteq_s \rho(\dr{k\,\overline{e}}{\boxempty})$.

\subsection{R5}

We have that $\rho(\dr{x\, \overline{e}}{\R}) =
\rho(\con{x\,\dr{\overline{e}}{\boxempty}})$, and the conditions of the
proposition ensure that $\fv(\con{x\, \overline{e}}) \cap dom(\rho) =
\emptyset$, so $\rho(\con{x\,\dr{\overline{e}}{\boxempty}}) =
\con{x\,\rho(\dr{\overline{e}}{\boxempty})}$. By the induction hypothesis,
$\overline{e} \costrighteq_s \rho(\dr{\overline{e}}{\boxempty})$, and by
congruence properties of strong improvement (Lemma
\ref{lem:simplaws}:1) $\con{x\,\overline{e}} \costrighteq_s
\rho(\dr{x\, \overline{e}}{\R})$. 

\subsection{R6} 

We have that $\rho(\dr{\lambda \overline{x}.e}{\boxempty}) = \rho(\lambda
\overline{x}.\dr{e}{\boxempty})$, and the conditions of the
proposition ensure that $\fv(\lambda \overline{x}.e) \cap  dom(\rho) =
\emptyset$, so $\rho(\lambda \overline{x}.\dr{e}{\boxempty}) = \lambda
\overline{x}.\rho(\dr{e}{\boxempty})$. By the induction hypothesis, $e
\costrighteq_s \rho(\dr{e}{\boxempty})$, and by
congruence properties of strong improvement (Lemma
\ref{lem:simplaws}:1) $\lambda \overline{x}.e \costrighteq_s
\rho(\dr{\lambda \overline{x}.e}{\boxempty})$. 

\subsection{R7}

We have that $\rho(\dr{n_1 \oplus n_2}{\R}) =
\rho(\dr{\con{n}}{\boxempty})$. By the induction hypothesis, $\con{n}
\costrighteq_s \rho(\dr{\con{n}}{\boxempty})$, and since $\con{n_1 \oplus
  n_2} \eval \con{n}$ it follows from Lemma \ref{lem:simplaws}:4 that
$\con{n_1 \oplus n_2} $ $\costrighteq_s \rho(\dr{n_1 \oplus n_2}{\R})$.

\subsection{R8}

\begin{enumerate}[$\bullet$]
  \item[a)] $e_1 \oplus e_2 = a$: 
We have that $\rho(\dr{e_1 \oplus e_2}{\R}) = \rho(\dr{e_1}{\boxempty}
  \oplus \dr{e_2}{\boxempty})$, by the given conditions 
 $\fv(\con{e_1 \oplus e_2}) \cap dom(\rho) = \emptyset$, so 
$\rho(\con{\dr{e_1}{\boxempty} \oplus \dr{e_2}{\boxempty}})  = $ \\ $
\con{\rho(\dr{e_1}{\boxempty}) \oplus \rho(\dr{e_2}{\boxempty})}$. By the induction
hypothesis $e_1 \costrighteq_s \rho(\dr{e_1}{\boxempty})$ and $e_2
\costrighteq_s \rho(\dr{e_2}{\boxempty})$, and by congruence properties of
strong improvement (Lemma \ref{lem:simplaws}:1) $\con{e_1 \oplus e_2}$ 
$\costrighteq_s\rho(\dr{e_1 \oplus e_2}{\R})$.
\item[b)] $e_1 = n$ or $e_1 = a$: We have $\rho(\dr{e_1 \oplus e_2}{\R}) = 
\rho(\dr{e_2}{\con{e_1 \oplus \boxempty}})$ and $\con{e_1 \oplus e_2}  \costrighteq_s
\rho(\dr{e_2}{\con{e_1 \oplus \boxempty}})$ follows  from the induction hypothesis.

\item[c)] otherwise: We have that $\rho(\dr{e_1 \oplus e_2}{\R}) = 
\rho(\dr{e_1}{\con{\boxempty \oplus e_2}})$ and $\con{e_1 \oplus e_2}
\costrighteq_s$ $\rho(\dr{e_1}{\con{\boxempty \oplus e_2}})$ follows  from the induction hypothesis.

\end{enumerate}
\subsection{R9}

We have that $\rho(\dr{(\lambda \overline{x}.f)\, \overline{e}}{\R}) =
\rho(\dr{\con{\LET \overline{x} = \overline{e} \IN
    f}}{\boxempty})$. Evaluating the input term  yields: $\con{(\lambda
\overline{x}.f)\, \overline{e}} \eval^r \con{(\lambda \overline{x}.f)\,
\overline{v}} \eval \con{[\overline{v}/\overline{x}]f}$, and evaluating the
input to the recursive call yields: $\con{\LET \overline{x} =
\overline{e} \IN f} \eval^r \con{\LET \overline{x} = \overline{v} \IN f}
\eval \con{[\overline{v}/\overline{x}]f}$. These two resulting terms are
syntactically equivalent, and therefore cost equivalent. By Lemma
\ref{lem:steps} their ancestor terms are cost equivalent, $\con{(\lambda
\overline{x}.f)\, \overline{e}} \costeq \con{\LET \overline{x} =
\overline{e} \IN f}$, and cost equivalence implies strong improvement.
By the induction hypothesis $\con{\LET \overline{x} = \overline{e} \IN
  f} \costrighteq_s \rho(\dr{\con{\LET \overline{x} = \overline{e} \IN
    f}}{\boxempty})$, and therefore $\con{(\lambda \overline{x}.f)\, \overline{e}}
\costrighteq_s \rho(\dr{(\lambda \overline{x}.f)\, \overline{e}}{\R})$.

\subsection{R10}
We have $\rho(\dr{e\,e'}{\R}) =
\rho(\dr{e}{\con{\boxempty\,e'}})$ and
$\con{e\,e'} \costrighteq_s
\rho(\dr{e}{\con{\boxempty\,e'}})$ follows from the induction
hypothesis.

\subsection{R11} 

We have that $\rho(\dr{\LET x = n \IN f}{\R}) =
\rho(\dr{\con{[n/x]f}}{\boxempty})$. By the induction hypothesis
$\con{[n/x]f} \costrighteq_s \rho(\dr{\con{[n/x]f}}{\boxempty})$, and since
$\con{\LET x = n \IN f} \eval \con{[n/x]f}$ it follows from Lemma
\ref{lem:simplaws}:4 that $\con{\LET x = n \IN f} \costrighteq_s
\rho(\dr{\LET x = n \IN f}{\R})$.

\subsection{R12}

We have that $\rho(\dr{\LET x = y \IN f}{\R}) =
\rho(\dr{\con{[y/x]f}}{\boxempty})$. By the induction hypothesis
$\con{[y/x]f} \costrighteq_s \rho(\dr{\con{[y/x]f}}{\boxempty})$, and since
$\con{\LET x = y \IN f} \costeq \con{[y/x]f}$ it follows that $\con{\LET x = y \IN f} \costrighteq_s
\rho(\dr{\LET x = y \IN f}{\R})$.

\subsection{R13} 

\subsubsection{Case: $x \in strict(f)$}

We have  $\rho(\dr{\LET x = e \IN f}{\R}) =
\rho(\dr{\con{[e/x]f}}{\boxempty})$. Evaluating the input term yields
$\con{\LET x = e \IN f} \eval^r \con{\LET x = v \IN f}$ $ \eval
\con{[v/x]f} \eval^s \econ{v}$, and evaluating the input to the
recursive call yields: $\con{[e/x]f} \eval^s \econ{e} \eval^r
\econ{v}$. These two resulting terms are
syntactically equivalent, and therefore cost equivalent. By Lemma
\ref{lem:steps} their ancestor terms are cost equivalent, $\con{\LET x
  = e \IN f} \costeq \con{[e/x]f}$, and cost equivalence implies
strong improvement. By the induction hypothesis $\con{[e/x]f}
\costrighteq_s \rho(\dr{\con{[e/x]f}}{\boxempty})$, and therefore $\con{\LET
  x = e \IN f} \costrighteq_s \rho(\dr{\LET x = e \IN f}{\R})$.

\subsubsection{Case: otherwise}

We have that $\rho(\dr{\LET x = e \IN f}{\R}) =$ \\ $ \rho(\LET x =
\dr{e}{\boxempty} \IN \dr{\con{f}}{\boxempty})$, and the conditions of the
proposition ensure that $\fv(\con{\LET x = e \IN f}) \cap dom(\rho) =
\emptyset$, so  $\rho(\LET x = \dr{e}{\boxempty} \IN \dr{\con{f}}{\boxempty}) =$\\ $ \LET
x = \rho(\dr{e}{\boxempty}) \IN \rho(\dr{\con{f}}{\boxempty})$. By the induction
hypothesis $e \costrighteq_s \rho(\dr{e}{\boxempty})$ and $\con{f}
\costrighteq_s \rho(\dr{\con{f}}{\boxempty})$. By Lemma \ref{lem:eqlet2} the
input is strongly improved by $\LET x = e \IN \con{f}$, and therefore
$\con{\LET x = e \IN f} \costrighteq_s \rho(\dr{\LET x = e \IN f}{\R})$.

\subsection{R14}

We have that $\rho(\dr{\LETREC g = v \IN e}{\R}) = \rho(\LETREC g =
v \IN \dr{\con{e}}{\boxempty})$, and the conditions of the
proposition ensure that $\fv(\con{\LETREC g = v \IN e}) \cap dom(\rho) =
\emptyset$, so $\rho(\LETREC g = v \IN \dr{\con{e}}{\boxempty}) = \LETREC
g = v \IN \rho(\dr{\con{e}}{\boxempty})$. By the induction
hypothesis  $\con{e}
\costrighteq_s \rho(\dr{\con{e}}{\boxempty})$. By Lemma \ref{lem:eqletrec} the
input is strongly improved by $\LETREC g = v \IN \con{e}$, and therefore
$\con{\LETREC g = v \IN e} \costrighteq_s$ $ \rho(\dr{\LETREC g = v \IN e}{\R})$.

\subsection{R15}

We have $\rho(\dr{\CASE x \OF \{ p_i \to e_i \}}{\R}) =
\rho(\CASE x \OF \{p_i \to \dr{\con{e_i}}{\boxempty}\})$, and the conditions
of the proposition ensure that $\fv(\con{\CASE x \OF \{ p_i \to e_i
  \}}) \cap dom(\rho) = \emptyset$, so $\rho(\CASE x \OF \{p_i \to
\dr{\con{e_i}}{\boxempty}\}) = \CASE x \OF \{p_i \to
\rho(\dr{\con{e_i}}{\boxempty})\}$. By the induction hypothesis $\con{e_i}
\costrighteq_s \rho(\dr{\con{e_i}}{\boxempty})$.  Using Lemma \ref{lem:eqcase}
the input is strongly improved by $\CASE x \OF \{p_i \to \con{e_i}
\}$, and therefore $\con{\CASE x \OF \{ p_i \to e_i \}} \costrighteq_s$ \\ $
\rho(\dr{\CASE x \OF \{ p_i \to e_i \}}{\R})$. 

\subsection{R16}

We have $\rho(\dr{\CASE k_j\,\overline{e} \OF \{ p_i \to e_i
  \}}{\R}) = \rho(\dr{\con{\LET \overline{x}_j = \overline{e} \IN
    e_j}}{\boxempty})$. \\ Evaluating the input term yields $\con{\CASE
  k_j\,\overline{e} \OF \{ p_i \to e_i \}} \eval^r \con{\CASE
  k_j\,\overline{v} \OF \{ p_i \to e_i \}} \eval
\con{[\overline{v}/\overline{x}_j]e_j}$, and evaluating the input to
the recursive call yields $\con{\LET \overline{x}_j = \overline{e} \IN
  e_j} \eval^r \con{\LET \overline{x}_j = \overline{v} \IN e_j} \eval
\con{[\overline{v}/\overline{x}_j]e_j}$. These two resulting terms are
syntactically equivalent, and therefore cost equivalent. By Lemma
\ref{lem:steps} their ancestor terms are cost equivalent, $\con{\CASE
  k_j\,\overline{e} \OF \{ p_i \to e_i \}} \costeq \con{\LET
  \overline{x}_j = \overline{e} \IN e_j}$, and cost equivalence
implies strong improvement. According to the induction hypothesis $\con{\LET
  \overline{x}_j = \overline{e} \IN e_j} \costrighteq_s $ \\ $
\rho(\dr{\con{\LET \overline{x}_j = \overline{e} \IN e_j}}{\boxempty})$, and
therefore $\con{\CASE k_j\,\overline{e} \OF \{ p_i \to e_i \}}
\costrighteq_s$ \\ $\rho(\dr{\CASE k_j\,\overline{e} \OF \{ p_i \to e_i
  \}}{\R})$.

\subsection{R17}

We have that $\rho(\dr{\CASE n_j \OF \{p_i \arr e_i \} }) =
\rho(\dr{\con{e_j}}{\boxempty})$. By the induction hypothesis $\con{e_j}
\costrighteq_s \rho(\dr{\con{e_j}}{\boxempty})$, and since $\con{\CASE n_j
  \OF \{p_i \arr e_i \}} \eval \con{e_j}$ it follows from Lemma
\ref{lem:simplaws}:4 that $\con{ \CASE n_j \OF \{p_i \arr e_i
  \}}\costrighteq_s$ \\ $ \rho(\dr{\CASE n_j \OF \{p_i \arr e_i \} }{\R})$.

\subsection{R18} 

We have that $\rho(\dr{\CASE a \OF \{ p_i \to e_i \}}{\R}) = $ \\ $ 
\rho(\CASE \dr{a}{\boxempty} \OF \{p_i \to \dr{\con{e_i}}{\boxempty}\})$, and the conditions
of the proposition ensure that $\fv(\con{\CASE a \OF \{ p_i \to e_i
  \}}) \cap dom(\rho) = \emptyset$, so $\rho(\CASE \dr{a}{\boxempty} \OF \{p_i \to
\dr{\con{e_i}}{\boxempty}\}) = \CASE \rho(\dr{a}{\boxempty}) \OF \{p_i \to
\rho(\dr{\con{e_i}}{\boxempty})\}$. By the induction hypothesis $a \costrighteq_s
  $ \\ $\rho(\dr{a}{\boxempty})$ and  $\con{e_i} \costrighteq_s \rho(\dr{\con{e_i}}{\boxempty})$ and by Lemma \ref{lem:eqcase}
the input is strongly improved by  $\CASE a \OF \{p_i \to \con{e_i}
\}$, and therefore $\con{\CASE a \OF \{ p_i \to e_i \}} \costrighteq_s $ \\ $
\rho(\dr{\CASE a \OF \{ p_i \to e_i \}}{\R})$. 

\subsection{R19}

We have that $\rho(\dr{\CASE e \OF \{ p_i \to e_i \}}{\R}) = 
\rho(\dr{e}{\con{\CASE \boxempty \OF \{p_i \to e_i\}}})$ and
$\con{\CASE e \OF \{ p_i \to e_i \}} \costrighteq_s
\rho(\dr{e}{\con{\CASE \boxempty \OF \{ p_i \to e_i \}}})$ follows from the
induction hypothesis.

\subsection{R20}
\label{sec:slutbev}

We have that $\rho(\dr{e}{\R}) = \rho(\con{e})$, and the conditions of
the proposition ensure that $\fv(\con{e}) \cap dom(\rho) = \emptyset$,
so $\rho(\con{e}) = \con{e}$ This is syntactically equivalent to the
input, and we conclude $\con{e} \costrighteq_s \rho(\dr{e}{\R})$.


\end{document}